\documentclass[11pt]{article}
\pdfoutput=1
\usepackage{bm}
\usepackage[table,xcdraw,svgnames,dvipsnames]{xcolor}

\usepackage{microtype}
\usepackage[utf8]{inputenc}
\usepackage[T1]{fontenc}
\usepackage{lmodern}

\usepackage{graphicx,fullpage,paralist}
\usepackage{geometry}
\usepackage{amsmath,amssymb,amsthm}
\usepackage{comment,hyperref}
\usepackage{subcaption}
\usepackage{thmtools}
\usepackage{xspace}
\usepackage{enumitem}
\usepackage[font=footnotesize]{caption}
\usepackage{bbold}
\usepackage{booktabs}
\usepackage{mathtools}
\usepackage{thm-restate}
\usepackage{multirow}
\usepackage[ruled,linesnumbered]{algorithm2e}
\usepackage{tikz}
\usepackage{pgfplots}
\pgfplotsset{compat=1.10}
\usepgfplotslibrary{fillbetween}
\usetikzlibrary{backgrounds,patterns,positioning}

\hypersetup{
	colorlinks,
	linkcolor={red!50!black},
	citecolor={blue!50!black},
	urlcolor={blue!80!black}
}
\SetKwInput{KwData}{Input}
\SetKwInput{KwResult}{Output}

\newcommand{\calM}{\mathcal{M}}
\newcommand{\calL}{\mathcal{L}}

\newcommand{\calE}{\mathcal{E}}
\newcommand{\calI}{\mathcal{I}}
\newcommand{\calB}{\mathcal{B}}

\newcommand{\N}{\mathbb{N}}

\newcommand{\FF}{\mathbb{F}}

\newcommand{\E}{\mathbb{E}}

\DeclareMathOperator{\Po}{Poisson}
\DeclareMathOperator{\rank}{rank}

\DeclareMathOperator{\supp}{supp}

\newcommand{\AUX}{\mathrm{AUX}}
\newcommand{\OPT}{\mathrm{OPT}}
\newcommand{\ALG}{\mathrm{ALG}}
\newcommand{\Greedy}{\mathrm{Greedy}}

\newcommand{\JMSP}{J\text{-MSP}\xspace}
\newcommand{\MSP}{MSP\xspace}

\newcommand{\Trans}{\mathrm{T}}

\newcommand{\Prob}[4]{P[#1,#2,#3,#4]}

\theoremstyle{plain}
\newtheorem{theorem}{Theorem}[section]
\newtheorem{lemma}[theorem]{Lemma}
\newtheorem{corollary}[theorem]{Corollary}
\newtheorem{proposition}[theorem]{Proposition}

\newtheorem{observation}[theorem]{Observation}

\theoremstyle{definition}
\newtheorem{definition}[theorem]{Definition}

\newtheorem{remark}[theorem]{Remark}



\title{The Multiple-Choice Matroid Secretary Problem}
 \author{Matías Ortiz-Angel\thanks{Department of Mathematical Engineering and Center for Mathematical Modeling CNRS-IRL 2807, Universidad de Chile.} \and José A. Soto\footnotemark[1]}
\date{}

\begin{document}
\maketitle

\vspace{1em}

\begin{abstract}
We introduce and study the multiple-choice matroid secretary problem, denoted
the $(J,\kappa)$-MSP. When the matroid has rank one and $\kappa=\infty$, the model is exactly the classical secretary problem with $J$ choices. Elements
arrive in uniformly random order. The algorithm may keep a candidate pool
$\AUX$, but the pool must remain feasible in the $J$-fold union matroid
$\calM^{(J)}$ and must satisfy $|\AUX|\le \kappa\cdot\rank(\calM)$. After all
arrivals, the algorithm returns the maximum-weight independent subset of
$\AUX$ in the original matroid $\calM$. This model separates online storage
from the final feasible solution. We study two natural ways to implement these multiple choices:
multi-track algorithms, which maintain $J$ independent sets of $\calM$, and
union-based algorithms, which maintain the pool directly in $\calM^{(J)}$.

Our main result is an exact optimal algorithm for transversal matroids in the
uncapacitated $(J,\infty)$ setting. For every fixed $J$, the optimal
probability-competitive ratio is exactly the optimal success probability of the
classical secretary problem with $J$ choices. Thus the rank-one transversal
matroid is already the worst case for the whole transversal class, and the
optimal guarantee converges exponentially fast to $1$ as $J$ grows. We also
analyze a natural single-threshold routing algorithm for capacitated
transversal matroids, where each right vertex has capacity $b$ and the global
pool has capacity $\kappa\cdot\rank(\calM)$. Although the rule is simple, its
analysis gives explicit finite-parameter bounds and asymptotic formulas. These
formulas show how the finite-rank loss caused by the global capacity constraint
decays with the number of right vertices, and how the local capacity $b$, the
number of tracks $J$, and the global capacity $\kappa$ interact. Finally, we
instantiate the multi-track approach for
$k$-column-sparse matroids, obtaining a guarantee $1-O(e^{-J/(ke)})$, and the
union-based approach for laminar matroids, obtaining a guarantee
$1-O(e^{-J/e})$.
\end{abstract}
\newpage

\section{Introduction and model}\label{sec:intro}

Consider an online system for assigning seats to clients. Clients arrive one at a
time, in uniformly random order. Each client has a value and a set of acceptable
seats. The system gains the value of a client only if it assigns that client an
acceptable seat at the end. The seat itself does not affect the value.

This is an online selection problem. If a client is rejected when it arrives, then it is lost
forever. We study a matroidal version of multiple choice: the algorithm may keep a
bounded candidate pool and choose the final feasible solution after all arrivals.

We control the pool in two ways. First, we fix an integer $J$. The pool must be
coverable by at most $J$ feasible allocations, which we call tracks. When a
client arrives, the system may add the client to the pool only if the current
pool together with that client is still coverable by at most $J$ tracks. In the
seating example, each track is a matching between clients and seats. Thus no
seat is assigned to more than one client in a single track. Across the whole
pool, at most $J$ clients can be associated with the same seat. Second, we may
impose a global capacity constraint on the total size of the pool.

After all clients have arrived, the system computes the final allocation from the
candidate pool. The track decomposition is only a certificate that the pool was
feasible. The system may discard it. It then matches a subset of the pooled
clients to the seats so as to maximize the total value.

This seating example is an online bipartite matching problem with weights on the
clients. A set of clients can be served exactly when it can be matched to
distinct seats. These feasible sets form a transversal matroid. Thus the example
is already a matroidal online selection problem, but with a delayed final choice.
We now state the corresponding model for an arbitrary matroid.

Throughout the paper, all weights are strictly positive and distinct.
Equivalently, the weights induce a total order on the ground set. For
$X\subseteq N$, let $\OPT(X)$ denote the unique maximum-weight independent set
contained in $X$.

\begin{definition}[The $(J, \kappa)$-Multiple-Choice Matroid Secretary Problem]
Let $\calM=(N,\calI)$ be a matroid of rank $r=\rank(\calM)\ge 1$. Let
$J\ge 1$ be an integer and let $\kappa\in(0,\infty]$ be a global capacity
parameter. Elements arrive in a uniformly random order, revealing their weights
upon arrival. The algorithm maintains a candidate pool $\AUX\subseteq N$,
initially empty. When an element $e$ arrives, the algorithm must either reject it
irrevocably or add it to $\AUX$. It may add $e$ only if $\AUX\cup\{e\}$ is
independent in the $J$-fold union matroid $\calM^{(J)}$ and
$|\AUX\cup\{e\}|\le \kappa r$. After all elements have arrived, the algorithm
returns $\ALG=\OPT(\AUX)$.

We abbreviate this problem as the $(J,\kappa)$-MSP. We call the case
$\kappa=\infty$ the $J$-MSP.
\end{definition}

Recall that a set $F\subseteq N$ is independent in $\calM^{(J)}$ if and only if
it can be written as the union of $J$ independent sets of $\calM$. Thus the
online phase builds a pool that can be covered by $J$ feasible solutions, while
the final phase extracts one independent set from this pool.

When $J=1$ and $\kappa=\infty$, this is the Matroid Secretary Problem
(MSP)~\cite{Babaioff07}. Indeed, the pool $\AUX$ is independent in $\calM$, and
all weights are positive. Hence $\OPT(\AUX)=\AUX$, so the final output is exactly
the set accepted online.

When $\calM$ has rank one and $\kappa=\infty$, the problem is the classical
secretary problem with $J$ choices: the algorithm may keep at most $J$ elements
and wins if the best one is kept.

There are two useful ways to maintain this pool. A \emph{multi-track algorithm}
keeps $J$ independent sets $F_1,\ldots,F_J$ of $\calM$ and routes each accepted
element to one of them. Its pool is $\AUX=F_1\cup\cdots\cup F_J$, so feasibility
in $\calM^{(J)}$ is certified by construction. A \emph{union-based algorithm}
keeps only the pool $\AUX$ and checks feasibility in $\calM^{(J)}$ directly. It
therefore does not commit online to a partition of $\AUX$ into tracks. We use
multi-track algorithms when the structure gives a good routing rule, as in
transversal and $k$-column-sparse matroids. We use union-based algorithms when
the union constraint itself is easy to enforce, as in laminar matroids. This
distinction is algorithmic; both families solve the same $(J,\kappa)$-MSP.

\paragraph{Relation to the Matroid Secretary Conjecture.}
The Matroid Secretary Conjecture asks whether the MSP admits an algorithm whose
expected output weight is a constant fraction of the optimum weight on every
matroid. The distinction above is relevant to this question. Suppose that, for
some absolute constant $J$, there is a multi-track algorithm for the $J$-MSP on
general matroids that returns every element of $\OPT(N)$ with probability at
least $c$. Choose an index $i\in[J]$ uniformly at random before the process
starts. Then run the multi-track algorithm and accept exactly the elements routed
to $F_i$. The set $F_i$ is independent in $\calM$, so this gives a feasible MSP
algorithm.

For every $e\in\OPT(N)$, the probability that $e$ is routed to the chosen track
is at least $c/J$, because each accepted element is placed in exactly one track.
Since all weights are positive, linearity of expectation gives
\[
    \mathbb{E}[w(F_i)]
    \ge \frac{c}{J}\,w(\OPT(N)).
\]
Thus such a multi-track algorithm, with $J$ constant, would imply the Matroid
Secretary Conjecture. This argument uses the maintained tracks. It does not
apply to union-based algorithms, because a pool feasible in $\calM^{(J)}$ need
not be given online as the union of $J$ maintained independent sets.

\subsection{Our results}

\paragraph{Optimality for transversal matroids.}
For transversal matroids, we first study the
$(J,\infty)$-MSP. If the right side of the bipartite graph has a single vertex, then the model
is exactly the $J$-choice secretary problem: the algorithm may keep up to $J$ elements and wins
if it keeps the maximum-weight element. This is the $J$-choice secretary problem,
initially studied by Gilbert and Mosteller~\cite{GilbertMosteller}. Optimal
multiple-threshold policies are known for this problem, and their success
probability converges to $1$ as $J$ grows.

Our first result extends this one-vertex benchmark to all transversal matroids.
We give a multi-track algorithm. At each right vertex, the algorithm uses a
multiple-threshold rule analogous to the optimal rule for the $J$-choice
secretary problem. The assignments are determined by the canonical matching of
the elements seen so far. At the end, the algorithm returns the maximum-weight
independent set contained in the pool.

For every fixed $J$, each element of $\OPT$ is returned with probability equal
to the optimal success probability of the $J$-choice secretary problem. In the
terminology of Section~\ref{sec:prelim}, the algorithm is
probability-competitive with that ratio. This is also the best possible uniform
guarantee for transversal matroids, since the one-vertex case is exactly the
$J$-choice secretary problem.

\paragraph{A single-threshold algorithm for capacitated transversal matroids.}

We next study transversal matroids with capacity $b$ on the right vertices. In
these matroids, a set of left vertices is independent if it can be covered by a
$(1,b)$-matching: each left vertex is matched once, and each right vertex is
matched at most $b$ times. When the right side has a single vertex, this is the
uniform matroid of rank $b$.

We study an extension of the single-threshold $1/e$-competitive algorithm of
Kesselheim, Radke, T\"onnis, and V\"ocking~\cite{KesselheimRTV13} for transversal matroids to the
$(J,\kappa)$-MSP on these capacitated transversal matroids. Let $q$ be the number
of right vertices; in the setting analyzed below, the rank is $r=qb$. At every
step, our algorithm computes the optimal $(1,b)$-matching among the elements seen
so far. If the newly arrived element belongs to this matching, the algorithm
assigns it to the corresponding right vertex and then places it in the first
track where this assignment remains feasible. The pool must also satisfy the
global capacity bound $|\AUX|\le \kappa qb$.

For this routing algorithm, we study how to choose the sampling threshold $p$ as a function of the parameters $(J,\kappa,b,q)$. We give explicit bounds for finite parameter values, and we also describe several asymptotic regimes. For $b=1$ and $\kappa=\infty$, the classical choice $p=1/e$ gives guarantee $1/e$ when $J=1$, and the same fixed threshold yields guarantee $1-1/e$ as $J\to\infty$. If $p$ is optimized as a function of $J$, the guarantee increases further; Theorem~\ref{thm:routing-asymptotic} gives the asymptotic error term $C_0 J^{-1/(2e)}e^{-J/e}$.

The pool must satisfy $|\AUX|\le \kappa qb$, so $\kappa$ bounds the pool size as a multiple of the rank. In our analysis, the probability of violating this bound decays exponentially in the rank parameter $qb$, via standard concentration bounds for Poisson random variables. Our analysis shows that moderate global capacity already improves the one-track guarantee. In the large-$b$ limit, when $J=2$, taking $\kappa$ only slightly larger than $1$ suffices to exceed $1-1/e$ after taking $q\to\infty$.

We also study the limit $b\to\infty$. For $J=1$, the single-threshold rule recovers the asymptotic $1-1/e$ guarantee of Greedy-Improving from B\'erczi, Livanos, Soto, and Verdugo~\cite{BercziLSV24} for large-rank uniform matroids. For every fixed $J$, the optimized limit of this rule is $1-e^{-J}$.

\paragraph{Other matroid classes.}

We then study other classes of matroids. For $k$-column-sparse matroids, we use
a multi-track algorithm: each track follows the canonical rows associated with
the optimal elements, and the union of the tracks forms the pool. This gives a
probability-competitive ratio of $1-O(e^{-J/(ke)})$. For graphic matroids, this
becomes $1-O(e^{-J/(2e)})$, because graphic matroids are $2$-column-sparse.

For laminar matroids, we use a union-based greedy algorithm. It does not maintain
an explicit decomposition into $J$ tracks. Instead, it keeps the pool feasible
directly in $\calM^{(J)}$, and achieves a probability-competitive ratio of
$1-O(e^{-J/e})$.

\subsection{Related work}

\paragraph{The Matroid Secretary Problem.}
The Matroid Secretary Problem (\MSP) was introduced by Babaioff, Immorlica, and
Kleinberg~\cite{Babaioff07}, who posed the Matroid Secretary Conjecture: every
matroid should admit an online algorithm with a constant competitive ratio. The
conjecture remains open. For general matroids, the best known guarantee is an
$O(\log\log \rank)$-competitive algorithm, obtained by Lachish~\cite{Lachish14}
and Feldman, Svensson, and Zenklusen~\cite{FeldmanSZ15}.

Constant-competitive algorithms are known for several important classes of
matroids. For laminar matroids, constant utility-competitive guarantees were
obtained by Im and Wang~\cite{ImW11} and by Jaillet, Soto, and
Zenklusen~\cite{JailletSZ13}. More recently, B\'erczi, Livanos, Soto, and
Verdugo~\cite{BercziLSV24} gave a $1-\ln 2$ probability-competitive algorithm
for laminar matroids. For transversal matroids, Kesselheim, Radke, T\"onnis, and
V\"ocking~\cite{KesselheimRTV13} gave a $1/e$ utility-competitive algorithm based
on bipartite matchings. The same algorithm was later shown to be $1/e$
probability-competitive through the framework of $1$-forbidden
algorithms~\cite{SotoTV18}.

\paragraph{Multiple-choice secretary problems.}
Rank-one instances are the reason for the name of our model: the \JMSP becomes
the dowry problem with $J$ choices of Gilbert and Mosteller~\cite{GilbertMosteller}, also known as the
$J$-choice secretary problem. Chan, Chen and Jiang~\cite{ChanChen} study the more general
$(J,K)$-secretary problem in a continuous-time model, where the algorithm makes
$J$ choices and the payoff is the expected number of selected items among the
$K$ best. The case $K=1$ is the one used in this paper. They prove that
multiple-threshold algorithms are optimal in this regime and relate the
continuous-time model to the finite model through a continuous linear programming
formulation. Matsui and Ano~\cite{matsui2016lower} study the related
multiple-stopping odds problem and its connection with multiple-stopping
secretary problems.

Our analysis also uses the labeling-scheme framework of B\'erczi, Livanos, Soto,
and Verdugo~\cite{BercziLSV24}. In particular, this framework gives the
independent-increments property used in the analysis of the transversal
algorithms.

\paragraph{Matroid unions.}
Recently, Gujjar, Hua, Kleinberg, and Qiu~\cite{GujjarHKQ26} studied the matroid
secretary problem when the matroid is the $k$-fold union of a given matroid. They
obtained a utility-competitive ratio of $1-O(\sqrt{\log(n)/k})$, which gives a
constant lower bound when $k\ge \log n$. It is natural to ask whether
$k\ge \log r$ might suffice, but obtaining a guarantee for an absolute constant
$k$ remains open.

Their setting is different from ours. In their model, the secretary instance is
defined directly on the $k$-fold union matroid. In our model, the online phase
builds a pool feasible in $\calM^{(J)}$, and the final output is an independent
set of the original matroid $\calM$. For multi-track algorithms, we also maintain
the $J$ independent tracks of $\calM$ online. This maintained decomposition is
the structure used in the reduction from the $J$-MSP back to the classical MSP.

\paragraph{Online bipartite matching with free disposal.}
The feasibility structure of transversal matroids is close to online bipartite
matching. However, the $J$-MSP differs from online matching with free
disposal~\cite{fahrbach2020edgeweighted, blanc2021multiway}. In the
free-disposal model, an online element is assigned to an offline vertex
immediately, and conflicts are resolved locally by discarding previously assigned
elements. In contrast, the $J$-MSP builds a candidate pool and computes the final
matching globally after all arrivals. Free-disposal algorithms are provably
bounded away from $1$ even for uniform weights and random
arrivals~\cite{chierichetti2025new}. This does not contradict our transversal
result, because that result uses the final global optimization over the pool.
\section{Preliminaries}\label{sec:prelim}

In this section, we describe the matroid classes we use, the continuous-time
coupling model, and the probability-competitiveness metric.

Throughout, we work in the \emph{online-representation model}. The ground set
size $n=|N|$ is known in advance. A concrete representation of the matroid exists,
but is revealed online through the arriving elements. The form of the
representation depends on the matroid class. For each class, we specify what is
known from the start and what is revealed when an element arrives.

A matroid $\calM=(N,\calI)$ consists of a finite ground set $N$ and a nonempty family $\calI\subseteq 2^N$ of \emph{independent sets} such that:
\begin{enumerate}
    \item (Downward closed) If $I\in\calI$ and $I'\subseteq I$, then $I'\in\calI$.
    \item (Augmentation) If $I,J\in\calI$ and $|I|<|J|$, then there exists $z\in J\setminus I$ such that $I\cup\{z\}\in\calI$.
\end{enumerate}
A maximal independent set is a \emph{basis}. All bases have the same size, called the \emph{rank} $r=\rank(\calM)$.

\subsection{Matroid classes, continuous-time coupling, and competitiveness}

We consider the following matroid classes and information models. In all cases, the
algorithm knows the ground set size $n=|N|$ in advance.
\begin{itemize}
    \item \textbf{General matroid.} An abstract matroid given only through an
    independence oracle. Upon each arrival, the algorithm may query the oracle on
    any subset of the elements revealed so far.

    \item \textbf{Transversal matroid.} There exists a bipartite graph
    $G=(N\cup R,E)$ such that $X\subseteq N$ is independent if $G$ has a matching
    that matches every vertex of $X$. The algorithm knows $R$ from the start, and
    when $v\in N$ arrives it learns its neighborhood $\Gamma(v)\subseteq R$.

    \item \textbf{Graphic matroid.} There exists an undirected graph $G=(V,E)$
    with $N=E$ such that $X\subseteq E$ is independent if it is acyclic. We allow
    parallel edges (i.e., $G$ may be a multigraph). The algorithm knows $V$ from
    the start, and when an edge arrives it learns its endpoints.

    \item \textbf{$k$-column-sparse (linear) matroid.} There exists a matrix
    representation whose columns are indexed by $N$, where each column has at
    most $k$ nonzero entries. The number of rows and the field are known from the
    start, and when $v\in N$ arrives the algorithm learns its entire column (i.e.,
    the list of its nonzero entries, including row indices and values).

    \item \textbf{Laminar matroid.} There exists a laminar family $\calL$ over $N$
    and capacities $\mu(C)$ such that $X\subseteq N$ is independent if
    $|X\cap C|\le \mu(C)$ for all $C\in\calL$. Our laminar algorithms only require
    access to an independence oracle.

\end{itemize}
\paragraph{Random-order model and continuous-time coupling.}
Elements of $N$ arrive in a uniformly random order. When an element $v$ arrives,
the algorithm observes its weight and the class-specific representation
information revealed at its arrival.

We represent the random order through a continuous-time coupling. As above, the
algorithm knows $n=|N|$. Before any arrivals, it draws i.i.d. samples
$u_1,\ldots,u_n\sim\mathrm{Unif}[0,1]$ and sorts them as
$0<u_{(1)}<\cdots<u_{(n)}\le 1$. When the $i$-th element arrives in the random
order, we assign it arrival time $t_v:=u_{(i)}$. This coupling induces a uniform
random permutation (ties occur with probability zero), and it lets the algorithm
use time-based threshold rules based on the assigned value $t_v$.

Throughout the paper, all weights are strictly positive and distinct. For
$X\subseteq N$, let $\OPT(X)$ denote the unique maximum-weight independent set
contained in $X$.

Since all weights are positive, $\OPT(X)$ is a basis of the restriction
$\calM|X$. In particular, if $X$ has rank $r=\rank(\calM)$, then
$|\OPT(X)|=r$. We write $\OPT:=\OPT(N)$. For $t\in(0,1]$, let
$N_t:=\{v\in N: t_v\le t\}$ and $\OPT_t:=\OPT(N_t)$.

In the $(J,\kappa)$-MSP, the online algorithm builds a candidate pool $\AUX\subseteq N$ that is independent in the $J$-fold union matroid $\calM^{(J)}$ and satisfies the global capacity bound. It then outputs $\ALG:=\OPT(\AUX)$.
A multi-track algorithm enforces feasibility by maintaining $\AUX$ as the union of $J$ independent sets of $\calM$. A union-based algorithm instead checks membership in $\calM^{(J)}$ directly.
In both cases, we compare the offline optimum with the final output $\ALG$, not with the pool $\AUX$. We make this comparison using probability competitiveness rather than the standard utility-competitive ratio.

\begin{definition}[Probability competitiveness]\label{def:probability-competitive}
Let $\OPT=\OPT(N)$. An online algorithm is
\emph{$c$-probability-competitive} if $\Pr(e\in\ALG)\ge c$ for every
$e\in\OPT$.
\end{definition}

When we state a guarantee for a matroid class $\mathfrak C$, the guarantee is
uniform over finite instances in the class. For an algorithm $A$, define
\[
    \rho_{\mathfrak C}(A)
    := \inf_{I\in\mathfrak C}\ \min_{e\in\OPT(I)} \Pr_A(e\in\ALG(I)),
\]
where the infimum ranges over all finite instances $I$ in $\mathfrak C$. The
optimal probability guarantee for the class is $\sup_A \rho_{\mathfrak C}(A)$.

For transversal matroids, we also track the right vertex that certifies an
insertion. A matching on $N_t$ has the total weight of its matched left vertices.
A canonical maximum-weight matching $M_t$ is a matching that maximizes this total
weight, with ties resolved by a fixed deterministic rule. Its matched left set is
$\OPT_t$. If an arriving element $e\in N_t$ belongs to $\OPT_t$, let $f_e\in R$
be the right vertex matched to $e$ in $M_t$. We say that $e$ \emph{chooses its
partner} $f_e$.

\subsection{Improving elements and the labeling scheme}

\begin{definition}[Improving elements]
An arriving element $e_t$ at time $t$ is \emph{improving} if $e_t\in\OPT_t$.
\end{definition}

The following fact lets us focus on improving elements.

\begin{proposition}[Folklore, see \cite{SotoTV18}]
\label{prop:folklore}
For every matroid $\calM=(N,\calI)$ and every $A\subseteq N$,
\[
    \OPT(N)\cap A \subseteq \OPT(A).
\]
In particular, every $e\in\OPT$ is improving at its arrival time.
\end{proposition}

\begin{proof}
Run greedy in decreasing order of weight; on $N$ it returns $\OPT(N)$, and on
$A$ it returns $\OPT(A)$. Fix $e\in\OPT(N)\cap A$. Since greedy accepts $e$ on
$N$, the elements of $N$ with weight larger than $w(e)$ do not span $e$. When
greedy is run on $A$, the elements considered before $e$ are exactly the elements
of $A$ with weight larger than $w(e)$, hence a subset of the larger-weight
elements of $N$. They do not span $e$, so greedy accepts $e$ on $A$. Thus
$e\in\OPT(A)$.
\end{proof}

\begin{observation}\label{obs:union-vs-output}
Let $U$ be the candidate pool produced by any execution of a $(J,\kappa)$-MSP
algorithm, and let $\ALG:=\OPT(U)$. Then, for every $v\in\OPT(N)$,
\[
    v\in U \iff v\in \ALG .
\]
\end{observation}

\begin{proof}
The implication $v\in \ALG \implies v\in U$ is immediate, since
$\ALG=\OPT(U)\subseteq U$. Conversely, suppose that $v\in U$ and $v\in\OPT(N)$.
By Proposition~\ref{prop:folklore}, applied with $A=U$, we get
\[
    v\in \OPT(N)\cap U \subseteq \OPT(U)=\ALG .
\]
This proves the equivalence.
\end{proof}
\paragraph{Full-rank convention.}\label{par:full-rank-convention}
In the algorithms analyzed in this paper, we use one technical convention. We fix
a constant $\tau_0\in(0,1)$ and assume that, for every
$t\ge \tau_0$, the set $N_t$ has rank $r$. Equivalently, $\OPT_t$ is a basis
for every $t\ge \tau_0$.

This convention is not part of the $(J,\kappa)$-MSP model. It is only used in the
analysis. The algorithms remain algorithms for the random-order model. The
continuous-time coupling only assigns times to the positions of a uniformly
random permutation.

All algorithms in this paper have a first sampling threshold $p>0$. Before time
$p$, they do not accept elements and do not change their state. After time $p$,
they act only on improving arrivals. If the arriving element is not improving,
the algorithm does nothing. In each application, we choose $\tau_0<p$. Thus
every arrival that can change the state of the algorithm occurs after the
full-rank convention is already valid.

We justify the convention using the construction of B\'erczi, Livanos, Soto, and
Verdugo~\cite{BercziLSV24}. They give the formal construction in their paper. We
recall only the part that we need. Add a dummy basis $X$ of size $r$. For each
$x\in X$, consider parallel dummy copies $x^1,x^2,x^3,\ldots$ with positive,
distinct, infinitesimal weights
$w(x^1)>w(x^2)>w(x^3)>\cdots$, all smaller than every original weight. Draw
independent arrival times for these copies. Let $x^{k(x)}$ be the first copy
whose arrival time lies in $(0,\tau_0)$. Keep $x^1,\ldots,x^{k(x)}$ and
discard all later copies. Then, by time $\tau_0$, at least one copy of each
$x\in X$ has arrived.

This does not make the input infinite. The infinite sequence of copies is only a
device for sampling the finite prefix $x^1,\ldots,x^{k(x)}$. This prefix is
finite almost surely. All later copies are discarded before the augmented
execution starts.

The discarded copies do not affect the process after time $\tau_0$. Each
discarded copy of $x$ is parallel to $x^{k(x)}$ and has a smaller weight. Since
$x^{k(x)}$ has already arrived by time $\tau_0$, no discarded copy can be
improving at any time $t\ge \tau_0$. Thus, removing the discarded copies does
not change the improving-arrival process on $[\tau_0,1]$.

The representation of $X$ depends on the matroid class: for transversal
matroids, we realize $X$ using universal left vertices; for graphic matroids, we
use the edges of a fixed spanning tree; and for $k$-column-sparse linear
matroids, we use the $r$ columns of the $r\times r$ identity matrix. For laminar
and general matroids, it suffices to add $r$ free dummy elements (and their
parallel copies) and truncate to rank $r$.

Dummy elements have positive, distinct, infinitesimal weights, all smaller than
every original weight. Thus they do not change the relative order of the original
elements, and they do not change the set of original elements in $\OPT$.

We analyze the algorithms on the augmented instance. Thus a dummy element that
arrives after the sampling threshold is processed by the algorithm exactly as any
other improving element. It may update the maintained sets, it may be inserted
into $\AUX$, and, when $\kappa<\infty$, it may consume global capacity. At the
end, we project the output back to the original instance by discarding all dummy
elements. All probability-competitive guarantees refer only to original elements
$e\in\OPT$.

After the augmentation, the elements seen by time $\tau_0$ contain one copy
of each element of $X$. Hence $N_t$ has rank $r$ for every $t\ge \tau_0$, and
$\OPT_t$ is a basis for every $t\ge \tau_0$. We treat the augmented instance
as the instance and suppress this distinction for the analysis.

This also does not make the guarantee asymptotic in $|N|$. The original instance
is finite, and the augmented execution uses only finitely many dummy elements
almost surely. The proof is carried out on this finite augmented instance, and
then the output is projected by discarding dummy elements. Thus the stated
probabilities concern original elements of $\OPT$ in the projected output.

The role of $\tau_0$ is only to avoid the singularity at time $0$. We apply the
Poisson property only to intervals contained in $[\tau_0,1]$ and after the first
sampling threshold.

\paragraph{Labeling schemes and the Poisson property.}\label{par:labeling-schemes}
A labeling scheme, in the sense of B\'erczi, Livanos, Soto, and
Verdugo~\cite{BercziLSV24}, is a rule that assigns labels to the elements of
$\OPT_t$ at each improving time $t\ge \tau_0$.
More precisely, for each such time $t$, the scheme specifies a bijection
\[
    \ell_t:\OPT_t\to [r].
\]
If the arriving element $e_t$ is improving, then the label of this improving
arrival is $\ell_t(e_t)$. This label is not chosen at random at time $t$; it is
the label that $e_t$ has under the bijection specified by the scheme for
$\OPT_t$.

The definition of B\'erczi, Livanos, Soto, and Verdugo also requires the scheme
to be independent of the internal arrival order inside $N_t$. Thus $\ell_t$ may
depend on $\OPT_t$, on the fixed data used by the scheme, and on future improving
arrivals and their times, but not on the internal arrival order inside $N_t$.

The continuous-time coupling only represents the finite random-order model. The
Poisson statements used below concern improving arrivals, and they come from the
labeling-scheme result of B\'erczi, Livanos, Soto, and Verdugo.

Under the full-rank convention, Lemma~3 of~\cite{BercziLSV24} gives the following
Poisson property.

\begin{proposition}[B\'erczi, Livanos, Soto, and Verdugo, see~\cite{BercziLSV24}]
\label{prop:poisson-label-general}
Assume the full-rank convention above. For any labeling scheme in the sense above, any subset of labels
$S\subseteq [r]$, and any interval $[s,t)$ with $\tau_0\le s<t\le 1$, let $N_S[s,t)$ denote the number of
improving arrivals in $[s,t)$ whose label lies in $S$:
\[
N_S[s,t) := \bigl|\{\tau\in[s,t) : e_\tau \text{ is improving and } \ell_\tau(e_\tau)\in S\}\bigr|.
\]
Then $N_S[s,t)\sim \Po\bigl(|S|\ln(t/s)\bigr)$. Moreover, the random variables $N_S[s,t)$ for disjoint
subintervals of $[\tau_0,1]$ are independent.

More generally, if $S_1,\ldots,S_m\subseteq[r]$ are
pairwise disjoint subsets of labels and $[s_1,t_1), \ldots,$ $[s_m,t_m)$ are intervals in $[\tau_0,1]$
(not necessarily disjoint), then $N_{S_1}[s_1,t_1),\ldots,N_{S_m}[s_m,t_m)$ are mutually independent.
\end{proposition}

We will use Proposition~\ref{prop:poisson-label-general} with labeling schemes
defined after fixing an element $e^*\in\OPT$. This is allowed by the definition
above, since the scheme may depend on the fixed target element and on future
improving arrivals, but not on the internal arrival order inside $N_t$.

We will also use the corresponding conditional form from the same labeling
scheme result: after conditioning on the arrival time $t_{e^*}=t$ of the fixed
target element, the label counts in intervals contained in $[\tau_0,t)$ have the
same Poisson distributions as in Proposition~\ref{prop:poisson-label-general}.

In all applications below, the first sampling threshold is larger than
$\tau_0$. Therefore, we only apply Proposition~\ref{prop:poisson-label-general}
to intervals after the sampling threshold, and these intervals are contained in
$[\tau_0,1]$.

\paragraph{Chernoff bounds for Poisson variables.}
We repeatedly upper bound Poisson tails using the following standard forms of the Chernoff (Cram\'er--Chernoff) bound.

\begin{lemma}\label{lem:poisson-chernoff}
Let $N\sim\Po(\mu)$.
\begin{enumerate}[label=(\alph*)]
    \item For every $x>\mu$,
    \[
        \Pr(N\ge x)
        \le
        \exp\bigl(-x\ln(x/\mu)+x-\mu\bigr)
        =
        \exp\bigl(-D(x\|\mu)\bigr),
    \]
    where $D(x\|\mu):=x\ln(x/\mu)-x+\mu$.
    \item For every integer $J\ge 1$ and every $0\le \lambda<J$,
    \[
        \Pr\bigl(\Po(\lambda)\ge J\bigr) \le \left(\frac{e\lambda}{J}\right)^J e^{-\lambda}.
    \]
\end{enumerate}
\end{lemma}
In the later algorithms, we will define concrete labeling schemes for the
matroid classes we study. In the next section, Lemma~\ref{lem:admissibility}
verifies that the labeling scheme used for transversal matroids is a labeling
scheme in the sense above.

Following B\'erczi, Livanos, Soto, and Verdugo~\cite{BercziLSV24}, we define the
\emph{improving word} on $[p,1)$ as follows. Fix $p\ge \tau_0$. Take the
labels of the improving arrivals in $[p,1)$ and read them in reverse
chronological order. This gives a word $z$ over the alphabet $[r]$. By
Proposition~\ref{prop:poisson-label-general} with $S=[r]$, the length of $z$ has
distribution $|z|\sim \Po(r\ln(1/p))$. Moreover, conditional on $|z|=m$, the word
$z$ is uniform over $[r]^m$.

\section{Transversal matroids and the \texorpdfstring{$J$}{J}-choice secretary problem}
\label{sec:transversal}

In this section, we give online algorithms for the $J$-MSP on transversal
matroids. The analysis reduces the selection probability of each fixed
$v\in\OPT$ to the success probability of the classical $J$-choice secretary
problem. This yields the optimal probability-competitive guarantee for this
class.

Throughout this section, let $\calM=(N,\calI)$ be a fixed transversal matroid of
rank $r$. We write $G=(N\cup R,E)$ for an underlying bipartite graph presentation;
the right vertex set $R$ is known from the start, and when an element $v\in N$
arrives the algorithm learns its neighborhood $\Gamma(v)\subseteq R$.

\paragraph{Matching assumption.}
In what follows, we assume that $G=(N\cup R,E)$ contains a matching that saturates
the right side $R$.

Under the online-representation model, the algorithm knows $R$ in advance but
learns the neighborhoods of left vertices only upon their arrival. Therefore,
rather than testing whether the assumption holds, we simply enforce it in the
analysis by augmenting the instance with (dummy) universal left vertices adjacent
to all of $R$. This guarantees a right-saturating matching and does not change
the transversal matroid induced on the original ground set $N$ (see the full-rank
convention, p.~\pageref{par:full-rank-convention}).

The matching assumption implies that the transversal rank is $r=|R|$ in the
augmented instance used for the analysis. By the full-rank convention, $\OPT_s$
is a basis of size $r$ for every $s\ge \tau_0$. Hence the canonical matching
$M_s$ saturates $\OPT_s$ and also saturates $R$. Therefore, for every right
vertex $f\in R$, and in particular for $f_v$, there is a unique element of
$\OPT_s$ matched to $f$ in $M_s$.

\paragraph{Labeling scheme for a single element in $\OPT$.}
Fix $v\in\OPT$ for the analysis. We define a labeling scheme (as defined in
Section~\ref{sec:prelim}) that is tailored to this fixed element $v$. As allowed
by the definition, this construction may depend on $v$ and on future improving
arrivals.

Let $t_v$ be the arrival time of $v$. Since $v\in\OPT(N)$,
Proposition~\ref{prop:folklore} implies that $v\in\OPT_{t_v}$. Let $M_{t_v}$ be
the canonical maximum-weight matching at time $t_v$, and let $f_v\in R$ be the
right vertex matched to $v$ in $M_{t_v}$.

For each improving time $s\ge \tau_0$, we assign label $1$ as follows.
If $s\ge t_v$, we set $\ell_s(v)=1$. If $s<t_v$, we set label $1$ to be the
(unique) element of $\OPT_s$ that is matched to $f_v$ in the canonical matching
$M_s$ (this element may be a dummy). The remaining labels are assigned
deterministically to the remaining elements of $\OPT_s$ in a way that is
independent of the internal arrival order of $N_s$.

\begin{lemma}[The rule above is a labeling scheme]
\label{lem:admissibility}
Fix $v\in\OPT$ and let $f_v$ be the right vertex matched to $v$ in the
canonical maximum-weight matching at time $t_v$. The rule defined above is a
labeling scheme in the sense used in Proposition~\ref{prop:poisson-label-general}.
\end{lemma}

\begin{proof}
At each improving time $s$, the rule assigns a bijection from $\OPT_s$ to
$[r]$. For $s<t_v$, label $1$ is assigned to the element of $\OPT_s$ matched
to the fixed right vertex $f_v$ in the canonical matching at time $s$. This
element exists and is unique by the matching assumption and the full-rank
convention. For $s\ge t_v$, label $1$ is assigned to $v$. The remaining labels
are assigned deterministically to the remaining elements of $\OPT_s$.

The canonical matching at time $s$ is determined by the set $N_s$, the weights,
and the fixed tie-breaking rule. Hence the label assignment is independent of
the internal arrival order inside $N_s$. The rule is defined with respect to
the fixed target element $v$, and it uses the associated quantities $t_v$ and
$f_v$ (defined from the canonical matching at time $t_v$). Note that, when
$s<t_v$, the definition of a labeling scheme permits the label assignment at time
$s$ to depend on information revealed after time $s$, including $t_v$ and $f_v$.
The rule may also refer to future improving times, as permitted by the definition of a labeling
scheme in Section~\ref{sec:prelim}. Thus the rule is a labeling scheme in the
required sense.
\end{proof}

By Lemma~\ref{lem:admissibility}, the rule above is a labeling scheme in the
sense of Proposition~\ref{prop:poisson-label-general}. In particular, the
resulting label-$1$ counting process has independent increments. Moreover, for
any improving time $s<t_v$, the unique label-$1$ element is precisely the element
of $\OPT_s$ whose canonical partner in $M_s$ is $f_v$. Thus, before $v$ arrives,
the label-$1$ arrivals are exactly the improving arrivals that compete with $v$
for the right vertex $f_v$.

Throughout, we choose the parameter $\tau_0$ in the full-rank convention
(see Paragraph~\ref{par:full-rank-convention}) to be smaller than the first time
threshold used by the algorithm under analysis (e.g., $\tau_0<p$ for
Algorithm~\ref{alg:single} and $\tau_0<p_1$ for Algorithm~\ref{alg:optJ}).
Applying Proposition~\ref{prop:poisson-label-general}
to $S=\{1\}$ gives
\[
N_{\{1\}}[s,t)\sim\Po(\ln(t/s))
\qquad\text{for every }\tau_0\le s<t\le 1.
\]

\subsection{Warm-up: \texorpdfstring{MSP}{MSP} on transversal matroids, the case \texorpdfstring{$J=1$}{J=1}}

We start with the case $J=1$, that is, the classical \MSP on transversal
matroids~\cite{KesselheimRTV13}. The algorithm maintains a single independent set
throughout the execution, and it may accept multiple elements. We analyze its
probability guarantee by fixing an element $v\in\OPT$. We then use the labeling
scheme from the previous paragraph to track the improving arrivals that compete
with $v$.

\begin{algorithm}[ht]
\caption{Single-threshold routing algorithm for transversal matroids}
\label{alg:single}
\DontPrintSemicolon
\KwData{A transversal matroid $\calM=(N,\calI)$ with right vertex set $R$, a threshold $p\in(0,1)$, and access to the arrival time $t_e$ and neighborhood $\Gamma(e)\subseteq R$ of each element $e$ upon its arrival.}
\KwResult{An independent set in $\calI$.}
$\ALG\leftarrow \emptyset$, $P\leftarrow \emptyset$\;
Reject all arrivals before time $p$\;
\ForEach{arriving element $e$ at time $t_e\in[p,1)$ (in increasing order of $t_e$)}{
compute the canonical maximum-weight matching $M_{t_e}$ on the revealed graph induced by $N_{t_e}\cup R$\;
\If{$e\in\OPT_{t_e}$}{
let $f_e$ be the canonical partner of $e$ in $M_{t_e}$\;
\lIf{$f_e\notin P$}{add $e$ to $\ALG$ and add $f_e$ to $P$}\tcp*[r]{Otherwise discard $e$.}
}}
\Return{$\ALG$}\;
\end{algorithm}

The returned set is independent. Each accepted element is assigned to a distinct
right vertex in its revealed neighborhood.

\begin{theorem}
Algorithm~\ref{alg:single} is probability-competitive with guarantee $p\ln(1/p)$. The maximum over $p$ is $1/e$, attained at $p=1/e$.
\end{theorem}

\begin{proof}

Fix $v\in\OPT$ and use the labeling defined with respect to $v$. Then $v$ is the
last arriving element that carries label $1$. Indeed, for every $t\ge t_v$, the
labeling assigns label $1$ to $v$ inside $\OPT_t$, so no arrival after $t_v$ can
carry label $1$.

Before $v$ arrives, the label-$1$ arrivals are exactly the improving arrivals
whose canonical partner is $f_v$. Algorithm~\ref{alg:single} accepts $v$ if and
only if $t_v\ge p$ and no earlier label-$1$ arrival in $[p,t_v)$ has already used
the right vertex $f_v$. Therefore, $v$ is accepted if and only if $v$ is the
unique label-$1$ arrival in $[p,1)$, i.e., $N_{\{1\}}[p,1)=1$.

By Proposition~\ref{prop:poisson-label-general},
$\Pr(v\in\ALG)=\Pr\bigl(N_{\{1\}}[p,1)=1\bigr)=p\ln(1/p)$.
This is maximized at $p=1/e$, with value $1/e$.
\end{proof}

\subsection{Optimal algorithm for the \texorpdfstring{$J$}{J}-MSP on transversal matroids}
\label{sec:multiple-threshold}
For rank-one transversal matroids, the $J$-MSP coincides with the classical
$J$-choice secretary problem~\cite{GilbertMosteller}.

We recall the Gilbert--Mosteller policy for the $J$-choice secretary problem.
Fix $J$ and thresholds $0<p_1<\cdots<p_J<1$. The corresponding
multiple-threshold policy is defined below. Gilbert and
Mosteller~\cite{GilbertMosteller} derived the optimal limiting thresholds by
analyzing the finite $n$-candidate problem and then taking the limit as
$n\to\infty$.

Let $\gamma_J$ denote the success probability achieved by the optimal threshold
choice in the continuous-time limit. Chan, Chen and Jiang~\cite{ChanChen} showed that
$\gamma_J$ is the optimal asymptotic value: for every $\varepsilon>0$, no
algorithm can guarantee more than $\gamma_J+\varepsilon$ for all sufficiently
large $n$, while the Gilbert--Mosteller policy attains $\gamma_J$ asymptotically.

To describe the policy, we use the continuous-time coupling and arrival times $t\in[0,1]$ from
Section~\ref{sec:prelim}. Call an arrival a \emph{record} if its weight exceeds all previous weights; the
best candidate is the last record. The goal is to accept this best candidate while selecting at most $J$
candidates overall.

Given thresholds $(p_1,\dots,p_J)$, set $p_{J+1}:=1$. The policy rejects every record arriving before $p_1$.
Then, for each $j\in[J]$, over the interval $[p_j,p_{j+1})$ it accepts a record whenever it has accepted fewer
than $j$ records so far.

Given a realization of the arrival process (arrival times and weights), let $c_j$ be the number of record
arrivals that fall in the interval $[p_j,p_{j+1})$, and write $c=(c_1,\dots,c_J)\in\N^J$ for the resulting
count vector. Whether the policy (run on this realization) ends up accepting the last record depends only on
$c$. Let $\calB_J\subseteq\N^J$ be the set of count vectors $c$ for which the policy accepts the last record.
The next proposition gives the distribution of this count vector under the continuous-time coupling.

\begin{proposition}[Gilbert--Mosteller record decomposition, see~\cite{GilbertMosteller}]
\label{prop:GM}
Let $C=(C_1,\dots,C_J)$ count the record arrivals in the intervals
$[p_j,p_{j+1})$ under the continuous-time coupling. Then $C_1,\dots,C_J$ are
independent and
\[
\Pr(C=c)=\prod_{j=1}^{J}\frac{p_j}{p_{j+1}}\frac{(\ln(p_{j+1}/p_j))^{c_j}}{c_j!}
\qquad\text{for every } c\in\N^J.
\]
As a consequence,
\[
\Pr(\text{win on the }J\text{-choice secretary problem})
=\sum_{c\in\calB_J}\Pr(C=c).
\]
\end{proposition}

\begin{proof}
We derive this from Proposition~\ref{prop:poisson-label-general}. The $J$-choice
secretary problem is the $J$-MSP on a uniform matroid of rank one. In this case,
records are exactly the improving arrivals.

Apply Proposition~\ref{prop:poisson-label-general} with a single label. It implies that record arrivals
form a Poisson process, and that the record counts on disjoint time intervals are independent. Therefore,
for each $j\in[J]$, the number of records in $[p_j,p_{j+1})$ is Poisson with mean $\ln(p_{j+1}/p_j)$; that is,
$C_j\sim\Po(\ln(p_{j+1}/p_j))$. Hence
\[
\Pr(C_j=c_j)=e^{-\ln(p_{j+1}/p_j)}
\frac{(\ln(p_{j+1}/p_j))^{c_j}}{c_j!}
=\frac{p_j}{p_{j+1}}\frac{(\ln(p_{j+1}/p_j))^{c_j}}{c_j!}.
\]
Because the intervals are disjoint, the variables $C_j$ are independent. The
joint formula follows by multiplying these probabilities. The final equality is
just the definition of $\calB_J$.
\end{proof}

Let $(p_1^*,\dots,p_J^*)$ be a threshold vector maximizing the win probability in
Proposition~\ref{prop:GM}, and define
\[
\gamma_J := \Pr(\text{win}) = \Pr(C\in\calB_J)
\]
under this choice. We will use $(p_1^*,\dots,p_J^*)$ later when instantiating our transversal-matroid algorithm.

We now return to transversal matroids. Fix a threshold vector $(p_1,\dots,p_J)$ and set $p_{J+1}:=1$.
Think of each right vertex $f\in R$ as having $J$ slots, numbered $1,\dots,J$. Slot $i$ of $f$ can be filled by
at most one accepted element whose assigned partner is $f$. Equivalently, for each $i\in[J]$ we maintain a set
$F_i$ consisting of all elements placed into slot $i$ of their partner; by construction, $F_i$ is a matching
into $R$.

At arrival time $t_e$, we first test whether $e$ is improving, i.e., whether $e\in\OPT_{t_e}$. If not, we reject
it. Otherwise, we compute its canonical partner $f_e$ in $M_{t_e}$ and let $j$ be such that $t_e\in[p_j,p_{j+1})$.
View the $J$ slots of $f_e$ as being controlled by a Gilbert--Mosteller-type rule local to $f_e$: during
$[p_j,p_{j+1})$, $f_e$ may accept an additional assignment only if it has used fewer than $j$ of its slots so far.
Equivalently, we accept $e$ iff $f_e$ has a free slot among $1,\dots,j$, placing $e$ into the smallest such slot.
(The set $P_i$ in Algorithm~\ref{alg:optJ} records which
right vertices have already used slot $i$.)

\begin{algorithm}[ht]
\caption{Threshold algorithm for transversal matroids}
\label{alg:optJ}
\DontPrintSemicolon
\KwData{A transversal matroid $\calM=(N,\calI)$ with right vertex set $R$, an integer $J\ge 1$, thresholds
$0<p_1<\cdots<p_J<1$ (with $p_{J+1} := 1$), and access to the arrival time $t_e$ and neighborhood
$\Gamma(e)\subseteq R$ of each element $e$ upon its arrival.}
\KwResult{A feasible output for the $J$-MSP.}
Initialize $F_1,\dots,F_J\leftarrow\emptyset$ and partner sets $P_1,\dots,P_J\leftarrow\emptyset$\;
For each arriving element $e$ (at time $t_e\in[0,1]$), in increasing order of $t_e$:\;
\Indp
\lIf{$t_e<p_1$}{discard $e$ and continue}
let $j$ be the index such that $t_e\in[p_j,p_{j+1})$\;
compute the canonical maximum-weight matching $M_{t_e}$ on the revealed graph induced by $N_{t_e}\cup R$\;
\If{$e\in\OPT_{t_e}$}{
let $f_e$ be the canonical partner of $e$ in $M_{t_e}$\;
\If{there exists $i\in\{1,\dots,j\}$ such that $f_e\notin P_i$}{
    let $i^*$ be the smallest index in $\{1,\dots,j\}$ such that $f_e\notin P_{i^*}$\;
    add $e$ to $F_{i^*}$ and add $f_e$ to $P_{i^*}$\;
}
}
\tcp*[f]{Otherwise discard $e$.}

\Indm
\Return{$\OPT(F_1\cup\cdots\cup F_J)$}\;
\end{algorithm}
The algorithm is feasible. Whenever we add an element $e$ to some $F_i$, we assign it to its canonical
partner $f_e\in\Gamma(e)$ revealed at its arrival. The set $P_i$ ensures that no two elements of $F_i$ are assigned to
the same right vertex. Thus the assigned partners form a matching of $F_i$ into $R$, and hence $F_i\in\calI$.
The returned set is independent because it is $\OPT(F_1\cup\cdots\cup F_J)$ in $\calM$.

We now analyze Algorithm~\ref{alg:optJ}. Fix an element $v\in\OPT$. Under the
labeling defined with respect to $v$, the label-$1$ arrivals are exactly the
improving arrivals that compete with $v$ for the fixed right vertex $f_v$.
Our first step is to compare this label-$1$ process with the record process in
the $J$-choice secretary problem.

\begin{lemma}
\label{lem:transfer-distribution}
Fix $v\in\OPT$ and use the labeling defined with respect to $v$. For each $j\in[J]$, let
\[
C^{\Trans}_j := N_{\{1\}}[p_j,p_{j+1})
\]
be the number of label-$1$ arrivals in $[p_j,p_{j+1})$, and set $C^{\Trans}=(C^{\Trans}_1,\dots,C^{\Trans}_J)$.
Then the coordinates of $C^{\Trans}$ are independent and, for every $c=(c_1,\dots,c_J)\in\N^J$,
\[
\Pr(C^{\Trans}=c)=\prod_{j=1}^{J}\frac{p_j}{p_{j+1}}\frac{(\ln(p_{j+1}/p_j))^{c_j}}{c_j!}.
\]
In particular, $C^{\Trans}$ has the same distribution as the record-count vector $C$ in Proposition~\ref{prop:GM}.
\end{lemma}

\begin{proof}
By Proposition~\ref{prop:poisson-label-general} with $S=\{1\}$, the number of label-$1$ arrivals in
$[p_j,p_{j+1})$ is Poisson with mean $\ln(p_{j+1}/p_j)$, i.e.,
$C^{\Trans}_j\sim\Po(\ln(p_{j+1}/p_j))$. Moreover, counts on disjoint intervals are independent, so the variables
$C^{\Trans}_1,\dots,C^{\Trans}_J$ are independent. Therefore,
\[
\Pr(C^{\Trans}_j=c_j)=e^{-\ln(p_{j+1}/p_j)}\frac{(\ln(p_{j+1}/p_j))^{c_j}}{c_j!}
=\frac{p_j}{p_{j+1}}\frac{(\ln(p_{j+1}/p_j))^{c_j}}{c_j!}.
\]
Multiplying over $j$ yields the stated joint distribution.
\end{proof}

We next compare decisions in the two settings: the $J$-MSP on a transversal matroid (Algorithm~\ref{alg:optJ})
and the classical $J$-choice secretary problem (the Gilbert--Mosteller policy). Recall that $\calB_J\subseteq\N^J$
is the set of count vectors $c$ for which the Gilbert--Mosteller policy accepts the final record; note that
$\calB_J$ depends only on $J$ and not on the particular values of $(p_1,\dots,p_J)$.
For a fixed threshold vector $(p_1,\dots,p_J)$, we show that, under a count vector $c$, Algorithm~\ref{alg:optJ}
places $v$ into $F_1\cup\cdots\cup F_J$ if and only if the Gilbert--Mosteller policy accepts the final record under
the same count vector $c$.

\begin{lemma}
\label{lem:transfer-decision}
Fix $v\in\OPT$ and condition on $C^{\Trans}=c$. Then Algorithm~\ref{alg:optJ}
places $v$ into some set $F_i$ if and only if $c\in\calB_J$.
\end{lemma}

\begin{proof}
For every improving time $s<t_v$, the label-$1$ arrival is exactly the improving
element whose canonical partner in $M_s$ is $f_v$. Hence, before $v$ arrives, the
only arrivals that can occupy slots of $f_v$ are the label-$1$ arrivals. Arrivals
after $t_v$ do not affect whether $v$ is placed into one of the sets
$F_1,\dots,F_J$, because the algorithm never removes an accepted element from a
track.

If $c_1=\cdots=c_J=0$, then there are no label-$1$ arrivals after time $p_1$.
Since $v$ is the last label-$1$ arrival, $v$ arrives before $p_1$. The final
record in the $J$-choice secretary problem also occurs before $p_1$, so both
policies reject.

Assume now that $c\ne 0$, and let $j^*:=\max\{j:c_j>0\}$. Since $v$ is the last
label-$1$ arrival, it lies in $[p_{j^*},p_{j^*+1})$. For $j<j^*$, let $A_j$ be
the number of slots of $f_v$ used by Algorithm~\ref{alg:optJ} by time
$p_{j+1}$, with $A_0=0$. During $[p_j,p_{j+1})$, only slots $1,\dots,j$ of
$f_v$ are available. Since all label-$1$ arrivals in intervals $j<j^*$ occur
before $v$, we have
\[
    A_j=\min\{j,A_{j-1}+c_j\}
    \qquad\text{for } j<j^* .
\]

Now consider the $J$-choice secretary problem under the same count vector $c$.
Let $R_j$ be the number of records accepted by the Gilbert--Mosteller policy by
time $p_{j+1}$, with $R_0=0$. During $[p_j,p_{j+1})$, the policy accepts records
until it has accepted $j$ records in total. Thus
\[
    R_j=\min\{j,R_{j-1}+c_j\}.
\]
Since $R$ and $A$ satisfy the same recurrence up to $j$, it follows that $A_j=R_j$ for every $j<j^*$.

During interval $[p_{j^*},p_{j^*+1})$, the element $v$ is the last label-$1$
arrival. Algorithm~\ref{alg:optJ} places $v$ into $F_1\cup\cdots\cup F_J$ if
and only if the $c_{j^*}$ label-$1$ arrivals in this interval, including $v$,
fit into the $j^*$ available slots of $f_v$. Equivalently,
\[
    A_{j^*-1}+c_{j^*}\le j^*.
\]
In the $J$-choice secretary problem, the final record also lies in
$[p_{j^*},p_{j^*+1})$, and the Gilbert--Mosteller policy accepts it if and only
if
\[
    R_{j^*-1}+c_{j^*}\le j^*.
\]
Since $A_{j^*-1}=R_{j^*-1}$, the two conditions are equivalent. Hence
Algorithm~\ref{alg:optJ} places $v$ into $F_1\cup\cdots\cup F_J$ if and only if
$c\in\calB_J$.
\end{proof}

We now combine Lemmas~\ref{lem:transfer-distribution} and~\ref{lem:transfer-decision} to obtain our main transfer
result: for any fixed threshold vector $(p_1,\dots,p_J)$, Algorithm~\ref{alg:optJ} matches the success probability
of the corresponding $J$-choice secretary threshold policy under the same continuous-time coupling.

\begin{theorem}
\label{thm:main-transfer}
Fix thresholds $(p_1,\dots,p_J)$ and run Algorithm~\ref{alg:optJ} with the continuous-time coupling from
Section~\ref{sec:prelim}. Then, for every $v\in\OPT$,
\[
\Pr(v\in\ALG)
= \Pr\bigl( C\in\calB_J \bigr)
= \Pr(\text{the corresponding $J$-choice threshold policy wins}).
\]
\end{theorem}
\begin{proof}
By Observation~\ref{obs:union-vs-output}, $v\in\ALG$ if and only if Algorithm~\ref{alg:optJ} places $v$ into
$F_1\cup\cdots\cup F_J$.
Fix $v\in\OPT$ and consider the count vector $C^{\Trans}$. By Lemma~\ref{lem:transfer-distribution},
$C^{\Trans}\stackrel{d}{=}C$. By Lemma~\ref{lem:transfer-decision}, Algorithm~\ref{alg:optJ} places $v$ into the
union if and only if $C^{\Trans}\in\calB_J$. Therefore
\begin{align*}
\Pr(v\in\ALG)
&= \sum_{c\in\calB_J}\Pr(C^{\Trans}=c)
= \sum_{c\in\calB_J}\Pr(C=c) \\
&= \Pr(\text{the $J$-choice threshold policy accepts the best candidate}).\qedhere
\end{align*}
\end{proof}

This term-by-term equivalence implies that, under the continuous-time coupling,
the optimal probability-competitive ratio for the \JMSP on transversal matroids
matches the optimal win probability of the $J$-choice secretary problem. Recall
that we denote this optimal win probability by $\gamma_J$.

\begin{corollary}
\label{cor:optimal}
Let $\gamma_J$ be the optimal win probability of the $J$-choice secretary
problem in the continuous-time model. Under the continuous-time coupling, the
optimal uniform probability guarantee for the $J$-MSP on transversal matroids is
$\gamma_J$. Moreover, Algorithm~\ref{alg:optJ}, run with an optimal $J$-choice
threshold vector, is $\gamma_J$-probability-competitive.
\end{corollary}

\begin{proof}
The lower bound follows from Theorem~\ref{thm:main-transfer} with an optimal
threshold vector for the $J$-choice secretary problem.

For the upper bound, restrict to complete rank-one transversal matroids, where
one right vertex is adjacent to every left vertex. On these instances, the
continuous-time coupling gives exactly the classical $J$-choice secretary
problem. Hence no algorithm can have a uniform probability guarantee larger than
$\gamma_J$ over all transversal matroids.
\end{proof}

For $J=1$, this recovers the known transfer from Dynkin's threshold rule to the
algorithm of Kesselheim, Radke, T\"onnis, and V\"ocking for transversal matroids. For every fixed $J$, the
worst case over transversal matroids is already witnessed by complete rank-one
instances.

For completeness, we describe the first values of the Gilbert--Mosteller
thresholds and recall how Chan, Chen and Jiang express their competitive ratios. Fix
$J$, and let $p^*(J)=(p^*_1(J),\ldots,p^*_J(J))$, with
$0<p^*_1(J)<\cdots<p^*_J(J)<1$, be the threshold vector in the order used in
this paper. We also use the reversed vector
$\widehat p(J)=(\widehat p_1(J),\ldots,\widehat p_J(J))
:=(p^*_J(J),\ldots,p^*_1(J))$.

Chan, Chen and Jiang~\cite[Theorem~1.2]{ChanChen} give a procedure that computes an
increasing sequence of rational numbers $\theta_1,\theta_2,\ldots$. In the
reversed order, $\widehat p_j(J)=e^{-\theta_j}$ for $j=1,\ldots,J$. Thus
$\widehat p(J)=(e^{-\theta_1},\ldots,e^{-\theta_J})$, while
$p^*(J)=(e^{-\theta_J},\ldots,e^{-\theta_1})$ in the order used in this paper.
Therefore, $\widehat p(J)$ is the prefix of $\widehat p(J+1)$, and $p^*(J)$ is
the suffix of $p^*(J+1)$.

The sequence starts with $\theta_1=1$, $\theta_2=3/2$, and
$\theta_3=47/24$. Hence
$p^*(1)=(e^{-1})$, $p^*(2)=(e^{-3/2},e^{-1})$, and
$p^*(3)=(e^{-47/24},e^{-3/2},e^{-1})\approx(0.1411,0.2231,0.3679)$.

Let $\gamma_J$ be the success probability obtained by the Gilbert--Mosteller
policy with $J$ choices. For the objective of selecting the best candidate,
Theorem~1.2 of Chan, Chen and Jiang gives
\[
    \gamma_J=\sum_{j=1}^J \widehat p_j(J)=\sum_{j=1}^J p^*_j(J).
\]
The numerical values in Table~\ref{tab:chan-chen} are taken from their
computation.

\begin{table}[ht]
\centering
\begin{tabular}{ccc}
\toprule
$J$ & Threshold exponent $\theta_J$ & Success probability $\gamma_J$ \\
\midrule
1 & $1$ & 0.367879 \\
2 & $3/2$ & 0.591010 \\
3 & $47/24$ & 0.732103 \\
4 & $\approx 2.3967$ & 0.823121 \\
5 & $\approx 2.8230$ & 0.882550 \\
6 & $\approx 3.2410$ & 0.921675 \\
7 & $\approx 3.6530$ & 0.947588 \\
8 & $\approx 4.0603$ & 0.964831 \\
\bottomrule
\end{tabular}
\caption{Gilbert--Mosteller thresholds for the $J$-choice secretary
problem~\cite{ChanChen}. Under our convention,
$p^*(J)=(e^{-\theta_J},e^{-\theta_{J-1}},\ldots,e^{-\theta_1})$. The table
lists the exponent $\theta_J$ of the new threshold added when passing from
$J-1$ choices to $J$ choices. The success probability is
$\gamma_J=\sum_{j=1}^J p^*_j(J)$.}
\label{tab:chan-chen}
\end{table}

\begin{remark}
Table~\ref{tab:chan-chen} lists the threshold exponents and the values
$\gamma_J$ for small $J$. In Section~\ref{sec:routing}, we study a simpler
single-threshold routing algorithm for transversal matroids and show that it
achieves a probability-competitive ratio of at least
$1 - O\bigl(J^{-1/(2e)} e^{-J/e}\bigr)$.
\end{remark}

\section{Single-threshold routing for capacitated transversal matroids}
\label{sec:routing}

We now consider transversal matroids with capacities on the right vertices.
The ground set stays the same, but a right vertex may be used more than once.
We start by defining the independence system and then show that it is a matroid.

The algorithm analyzed in this section uses one sampling threshold. After the
threshold, it routes each improving element to the first track whose local
capacity at the canonical partner is not full. This rule is less refined than
the optimal multi-threshold algorithm from
Section~\ref{sec:multiple-threshold}, but its local routing decisions can be
analyzed in the presence of right-vertex capacities and a global bound on the
candidate pool.

\subsection{Capacitated transversal matroids}

Let $G=(N\cup R,E)$ be a bipartite graph, where $N$ is the ground set. We fix an
integer capacity $b\ge 1$ and give every right vertex $f\in R$ capacity $b$. Let
$q:=|R|$. As in the uncapacitated case, we assume that $R$ and $b$
are known from the start, and when an element $v\in N$ arrives the algorithm
learns its neighborhood $\Gamma(v)\subseteq R$.

A set $X\subseteq N$ is independent if we can assign each element of $X$ to an
adjacent right vertex, so that each element is assigned once and each right
vertex is assigned at most $b$ elements. Equivalently, the induced graph on
$X\cup R$ has a $(1,b)$-matching that covers $X$.

This independence system is a matroid. One way to see this is to replace each
right vertex $f\in R$ by $b$ copies $f^1,\dots,f^b$, each with the same
neighbors as $f$. Then $X$ has a covering $(1,b)$-matching in the original graph
if and only if $X$ has a matching in the copied graph. Thus the independent sets
are exactly those of a (standard) transversal matroid.

We use the full-rank convention from Section~\ref{sec:prelim}. Accordingly, for the
analysis, we may analyze the case
in which the capacitated transversal matroid has rank $r=bq$. Thus, for every
$t\ge \tau_0$, the set $N_t$ has rank $r$, and $\OPT_t$ has size $bq$.
Equivalently, $\OPT_t$ has size $bq$ for every $t\ge \tau_0$, and the
canonical maximum-weight $(1,b)$-matching on $N_t$ saturates the total
right-side capacity.

We analyze the algorithm on this augmented instance. Dummy elements have
negligible weights and are processed as ordinary elements: if a dummy element
is improving after the sampling threshold, the algorithm may accept it, route it
to a track, and count it against the global capacity. At the end we discard all
dummy elements and keep only the original elements. All probability guarantees
below concern original elements of $\OPT$.

With this convention, the global capacity bound is $|\AUX|\le \kappa r=\kappa bq$.
\subsection{The \texorpdfstring{$J$}{J}-track routing algorithm}

Throughout this section, a canonical maximum-weight $(1,b)$-matching means a
maximum-weight $(1,b)$-matching with ties broken by a fixed deterministic rule.
Algorithm~\ref{alg:routingJ-bcap} gives pseudocode for the routing algorithm.

\begin{algorithm}[ht]
\caption{Routing algorithm for capacitated transversal matroids}
\label{alg:routingJ-bcap}
\DontPrintSemicolon
\KwData{A $b$-capacitated transversal matroid with right vertex set $R$, a
global capacity parameter $\kappa>0$, an integer $J\ge 1$, a threshold
$p\in(0,1)$, and access to the arrival time $t_e$ and neighborhood
$\Gamma(e)\subseteq R$ of each element $e$ upon its arrival.}
\KwResult{A feasible output for the $(J,\kappa)$-MSP.}

Initialize $F_1,\dots,F_J\leftarrow\emptyset$ and usage counters
$U_1,\dots,U_J$, with $U_i(f)=0$ for all $i\in[J]$ and $f\in R$\;
Initialize $\AUX\leftarrow\emptyset$ and $A\leftarrow 0$\;
Set $K\leftarrow\infty$ if $\kappa=\infty$, and
$K\leftarrow\lfloor\kappa r\rfloor$ otherwise\;
Reject all arrivals before time $p$\;

For each arrival time $t\in[p,1)$ in increasing order:\;
\Indp
let $v_t$ be the arriving element at time $t$\;
compute the canonical maximum-weight $(1,b)$-matching $M_t$ on $N_t$\;
\If{$v_t\in\OPT_t$ and $A<K$}{
    let $f_t$ be the right vertex matched to $v_t$ in $M_t$\;
    \If{there exists $i\in\{1,\dots,J\}$ such that $U_i(f_t)<b$}{
        let $i^*$ be the smallest index in $\{1,\dots,J\}$ such that $U_{i^*}(f_t)<b$\;
        add $v_t$ to $F_{i^*}$ and to $\AUX$\;
        increment $U_{i^*}(f_t)$\;
        set $A\leftarrow A+1$\;
    }
}
\Indm

\Return{$\OPT(\AUX)$}
\end{algorithm}

Algorithm~\ref{alg:routingJ-bcap} is a single-threshold routing algorithm. It
maintains $J$ independent sets $F_1,\dots,F_J$ of the $b$-capacitated transversal
matroid and their union
\[
    \AUX := F_1\cup\cdots\cup F_J .
\]
Each right vertex $f\in R$ has capacity $b$ in each set $F_i$, so it has $bJ$
available slots across the $J$ sets.

At each time $t\ge p$, the algorithm computes the canonical maximum-weight
$(1,b)$-matching on $N_t$, with ties broken by a fixed deterministic rule. If
the arriving element $v_t$ is improving, let $f_t$ be its matched right vertex.
The algorithm routes $v_t$ to the first set $F_i$ in which $f_t$ has been used
fewer than $b$ times. It also enforces the global capacity bound through $K$: if
$A=K$, then $v_t$ is rejected even if a local slot is available. When
$\kappa=\infty$, we have $K=\infty$, and this test never rejects an element.

The algorithm is correct. Whenever an element is added to $F_i$, it is assigned
to its canonical partner $f_t\in\Gamma(v_t)$. The counter $U_i(f)$ ensures that
at most $b$ elements of $F_i$ are assigned to each right vertex $f$. Hence each
$F_i$ is independent in the $b$-capacitated transversal matroid. Therefore
$\AUX=F_1\cup\cdots\cup F_J$ is independent in the $J$-fold union matroid. The
counter $A$ enforces $|\AUX|\le K$. Thus, when $\kappa<\infty$, the pool
satisfies $|\AUX|\le\lfloor\kappa r\rfloor$; when $\kappa=\infty$, no global
capacity bound is imposed. In both cases, the returned set $\OPT(\AUX)$ is
independent in the original matroid.

For the probability-competitive analysis, it is enough to study whether a fixed
original element $v\in\OPT$ enters $\AUX$. Indeed, by
Observation~\ref{obs:union-vs-output}, this is equivalent to
$v\in\OPT(\AUX)$.

The rest of this section analyzes Algorithm~\ref{alg:routingJ-bcap}. The
parameters are $J$, $\kappa$, $b$, and $q$. We write $\lambda:=\ln(1/p)$.
For a fixed original element $v^*\in\OPT$, let
$\Prob{J}{\kappa}{b}{q}(\lambda)$ denote the probability that $v^*$ belongs to
the final output when the algorithm uses threshold $p=e^{-\lambda}$. We first
analyze the case $\kappa=\infty$. In that regime we obtain a formula for general
$b$, optimize it for $b=1$, and study the pointwise limit as $b\to\infty$. The
following subsections then add the global capacity constraint, first in the
limit $q\to\infty$ and then for finite $q$.

\paragraph{Block labeling scheme.}
We use a labeling scheme that extends the one used for standard transversal
matroids. In the standard case, one label follows each right vertex. Here each
right vertex has capacity $b$, so we use $b$ labels for each right vertex.

Fix a partition of $[r]=[bq]$ into blocks $(S_f)_{f\in R}$, with $|S_f|=b$ for
every $f\in R$. We now define the label assigned to each element of $\OPT_s$ at
an improving time $s\ge \tau_0$.

At time $s$, the canonical matching $M_s$ covers $\OPT_s$. Since $|\OPT_s|=bq$
and the total right-side capacity is $bq$, every right vertex is used to full
capacity. Thus, for each $f\in R$, exactly $b$ elements of $\OPT_s$ are matched
to $f$.

For each $f\in R$, we give the labels in $S_f$ to the elements of $\OPT_s$
matched to $f$, using a fixed deterministic order. This defines a bijection
\[
    \ell_s:\OPT_s\to[bq].
\]
The matching $M_s$ is determined by $N_s$, the weights, and the fixed
tie-breaking rule. The deterministic order is also fixed in advance. Hence the
rule does not depend on the internal arrival order inside $N_s$. Therefore,
$(\ell_s)_{s\ge \tau_0}$ is a labeling scheme in the sense of
Section~\ref{par:labeling-schemes}.

\subsection{No global capacity: \texorpdfstring{$\kappa=\infty$}{kappa=infinity}}
\label{sec:ratio-gap}

In this subsection, the global capacity constraint is absent. We first analyze a
fixed original element $v^*\in\OPT$. The argument below shows that its selection
probability depends on $J$, $b$, and the threshold $p$ only through the
parameter $\lambda:=\ln(1/p)$. We therefore write
\[
    P_{J,b}(\lambda):=\Prob{J}{\kappa=\infty}{b}{q}(\lambda).
\]

Let $t^*$ be the arrival time of $v^*$, and let $f^*$ be the right vertex
matched to $v^*$ in $M_{t^*}$. For this fixed element, we use the block-labeling
construction with a distinguished block $B^*=\{1,\ldots,b\}$ assigned, at each
improving time $s$, to the $b$ elements of $\OPT_s$ matched to $f^*$ in $M_s$.
This is a labeling scheme by the same argument as above: the full-rank
convention gives exactly $b$ such elements at each time, and the assignment is
determined by $N_s$, the weights, the canonical matching, and fixed deterministic
rules. The definition of a labeling scheme permits the rule to depend on the
fixed target element $v^*$ and on the associated vertex $f^*$.

Before time $t^*$, an improving arrival has a label in $B^*$ exactly when its
canonical partner is $f^*$. Hence, by
Proposition~\ref{prop:poisson-label-general}, conditional on $t^*=t$,
\[
    N_{B^*}[p,t)\sim \Po(b\ln(t/p)).
\]

Since $\kappa=\infty$, there is no global capacity constraint. Thus the
algorithm accepts $v^*$ exactly when the local capacity at $f^*$ is not full.
This happens if and only if at most $bJ-1$ earlier improving arrivals have
appeared with labels in $B^*$. Equivalently,
\[
    N_{B^*}[p,t^*)\le bJ-1.
\]
Integrating over $t^*$ gives the formula below.

\begin{lemma}
\label{lem:local_prob}
For every $\lambda > 0$,
\[
    P_{J,b}(\lambda)
    =
    e^{-\lambda}\int_0^\lambda e^z
    \Pr\left(\Po(bz) \le bJ-1\right)\,dz.
\]
In particular, for $b=1$,
\[
    P_{J,1}(\lambda)
    =
    e^{-\lambda}\sum_{k=1}^J\frac{\lambda^k}{k!}.
\]
For $b>1$,
\[
    P_{J,b}(\lambda)
    =
    e^{-\lambda}
    \sum_{k=1}^{bJ}
    \frac{b^{k-1}}{(b-1)^k}
    \Pr\left(\Po\left((b-1)\lambda\right)\ge k\right).
\]
\end{lemma}

\begin{proof}
By the block-labeling argument preceding the lemma and
Proposition~\ref{prop:poisson-label-general}, conditional on
$t^*=t\in[p,1]$, the number of earlier improving arrivals that can occupy slots
of $f^*$ has distribution $\Po(b\ln(t/p))$. Hence the acceptance probability of
$v^*$ is
\[
    \Pr\left(\Po(b\ln(t/p))\le bJ-1\right).
\]
Therefore,
\[
    P_{J,b}(\lambda)
    =
    \int_p^1
    \Pr\left(\Po(b\ln(t/p))\le bJ-1\right)\,dt .
\]
Put $p=e^{-\lambda}$ and set $z=\ln(t/p)$. Then $t=pe^z$ and
$dt=e^{-\lambda}e^z\,dz$. As $t$ ranges over $[p,1]$, $z$ ranges over
$[0,\lambda]$. Hence
\[
P_{J,b}(\lambda)
=
e^{-\lambda}\int_0^\lambda e^z
\Pr\left(\Po(bz)\le bJ-1\right)\,dz
=
e^{-\lambda}
\sum_{m=0}^{bJ-1}
\frac{b^m}{m!}
\int_0^\lambda e^{-(b-1)z}z^m\,dz.
\]
For $b=1$, this gives
\[
P_{J,1}(\lambda)
=
e^{-\lambda}
\sum_{m=0}^{J-1}
\frac{1}{m!}\int_0^\lambda z^m\,dz
=
e^{-\lambda}
\sum_{k=1}^J \frac{\lambda^k}{k!}.
\]
For $b>1$, use the substitution $x=(b-1)z$:
\[
\int_0^\lambda e^{-(b-1)z}\frac{z^m}{m!}\,dz
=
\frac{1}{(b-1)^{m+1}}
\int_0^{(b-1)\lambda} e^{-x}\frac{x^m}{m!}\,dx.
\]
By the Poisson-Gamma duality,
\[
    \int_0^{s}e^{-x}\frac{x^m}{m!}\,dx
    =
    \Pr\left(\Po(s)\ge m+1\right).
\]
Therefore
\[
\int_0^\lambda e^{-(b-1)z}\frac{z^m}{m!}\,dz
=
\frac{\Pr\left(\Po((b-1)\lambda)\ge m+1\right)}{(b-1)^{m+1}}.
\]
Substituting into the sum and reindexing with $k=m+1$ yields
\[
P_{J,b}(\lambda)
=
e^{-\lambda}
\sum_{k=1}^{bJ}
\frac{b^{k-1}}{(b-1)^k}
\Pr\left(\Po\left((b-1)\lambda\right)\ge k\right). \qedhere
\]
\end{proof}

Lemma~\ref{lem:local_prob} gives a one-dimensional expression for every $b$.
For $b=1$, the formula becomes explicit enough to optimize in closed form.

\paragraph{The case $b=1$.}

Assume $b=1$. We first choose the value of $\lambda$ that maximizes
$P_{J,1}(\lambda)$ over $\lambda>0$. By Lemma~\ref{lem:local_prob},
$P_{J,1}(\lambda)=e^{-\lambda}\sum_{k=1}^J\lambda^k/k!$. Differentiating gives
\[
    P_{J,1}'(\lambda)
    =
    e^{-\lambda}\left(1-\frac{\lambda^J}{J!}\right).
\]
Thus $P_{J,1}$ increases for $\lambda<(J!)^{1/J}$ and decreases for
$\lambda>(J!)^{1/J}$. Therefore the maximum over $\lambda>0$ is attained at
$\lambda^*=(J!)^{1/J}$. The next theorem estimates the error at this value.

\begin{theorem}
\label{thm:routing-asymptotic}
Let $\lambda^*=(J!)^{1/J}$. Then, as $J\to\infty$,
\[
    1-P_{J,1}(\lambda^*)
    =
    C_0 J^{-1/(2e)}e^{-J/e}(1+o(1)),
    \qquad
    \text{where } C_0=\frac{e}{e-1}(2\pi)^{-1/(2e)}.
\]
\end{theorem}

\begin{proof}
Set $\lambda=\lambda^*$ and $u=\lambda/J$. By
Lemma~\ref{lem:local_prob},
$P_{J,1}(\lambda)=e^{-\lambda}\sum_{k=1}^J\lambda^k/k!$. Since
$e^\lambda=\sum_{k\ge 0}\lambda^k/k!$, and since $\lambda^J=J!$, we get
\begin{align*}
    1-P_{J,1}(\lambda)
    &=
    e^{-\lambda}
    \left(
        1+\sum_{k\ge J+1}\frac{\lambda^k}{k!}
    \right) \\
    &=
    e^{-\lambda}
    \left(
        1+\sum_{h\ge 1}
        \frac{\lambda^{J+h}}{(J+h)!}
    \right) \\
    &=
    e^{-\lambda}
    \left(
        1+\sum_{h\ge 1}
        \frac{u^h}{\prod_{i=1}^h(1+i/J)}
    \right).
\end{align*}

We estimate the series in the last expression. Stirling's formula gives
$u=(J!)^{1/J}/J\to 1/e$ as $J\to\infty$, so $u\le 1/2$ for all large $J$.
For each fixed $h\ge 1$, as $J\to\infty$,
\[
    \frac{u^h}{\prod_{i=1}^h(1+i/J)}
    \longrightarrow e^{-h},
    \qquad
    0\le
    \frac{u^h}{\prod_{i=1}^h(1+i/J)}
    \le u^h\le 2^{-h}.
\]
Dominated convergence gives, as $J\to\infty$,
\[
    1+\sum_{h\ge 1}
    \frac{u^h}{\prod_{i=1}^h(1+i/J)}
    =
    1+\sum_{h\ge 1}e^{-h}+o(1)
    =
    \frac{e}{e-1}(1+o(1)).
\]

It remains to estimate $e^{-\lambda}$. Stirling's formula gives
\begin{align*}
    J!
    &=
    \sqrt{2\pi J}
    \left(\frac{J}{e}\right)^J
    \exp\left(\frac{1}{12J}+O(J^{-3})\right), \\
    \lambda=(J!)^{1/J}
    &=
    \frac{J}{e}
    \exp\left(
        \frac{\ln(2\pi J)}{2J}
        +\frac{1}{12J^2}
        +O(J^{-4})
    \right) \\
    &=
    \frac{J}{e}
    +\frac{\ln(2\pi J)}{2e}
    +o(1),
\end{align*}
as $J\to\infty$. Therefore, as $J\to\infty$,
\[
    e^{-\lambda}
    =
    \exp\left(
        -\frac{J}{e}
        -\frac{\ln(2\pi J)}{2e}
        +o(1)
    \right)
    =
    e^{-J/e}(2\pi J)^{-1/(2e)}(1+o(1)).
\]
The estimates for the series and for $e^{-\lambda}$, together with
$\lambda=\lambda^*$, yield
\begin{align*}
    1-P_{J,1}(\lambda^*)
    &=
    e^{-J/e}(2\pi J)^{-1/(2e)}
    \frac{e}{e-1}(1+o(1)) \\
    &=
    \frac{e}{e-1}(2\pi)^{-1/(2e)}
    J^{-1/(2e)}e^{-J/e}(1+o(1)) \\
    &=
    C_0J^{-1/(2e)}e^{-J/e}(1+o(1)).\qedhere
\end{align*}
\end{proof}

For small values of $J$, the optimized ratio has the exact form
\[
    P_{J,1}\left((J!)^{1/J}\right)
    =
    e^{-(J!)^{1/J}}
    \sum_{k=1}^J\frac{(J!)^{k/J}}{k!}.
\]
Table~\ref{tab:exact-ratios} gives the first values, rounded down to three
decimal places, together with the limit as $J\to\infty$.

\begin{table}[ht]
\centering
\begin{tabular}{|c|p{10cm}|c|}
\hline
$J$ & \multicolumn{1}{c|}{$P_{J,1}((J!)^{1/J})$} & Lower bound \\
\hline
$1$ & $e^{-1}$ & $0.367$ \\[4pt]
$2$ & $e^{-\sqrt{2}}(1+\sqrt{2})$ & $0.586$ \\[4pt]
$3$ & $e^{-6^{1/3}}\left(1+6^{1/3}+\dfrac{6^{2/3}}{2}\right)$ & $0.726$ \\[8pt]
$4$ & $e^{-24^{1/4}}\left(1+24^{1/4}+\dfrac{24^{1/2}}{2}+\dfrac{24^{3/4}}{6}\right)$ & $0.816$ \\[8pt]
$5$ & $e^{-120^{1/5}}\left(1+120^{1/5}+\dfrac{120^{2/5}}{2}+\dfrac{120^{3/5}}{6}+\dfrac{120^{4/5}}{24}\right)$ & $0.876$ \\[8pt]
$6$ & $e^{-720^{1/6}}\left(1+720^{1/6}+\dfrac{720^{2/6}}{2}+\dfrac{720^{3/6}}{6}+\dfrac{720^{4/6}}{24}+\dfrac{720^{5/6}}{120}\right)$ & $0.916$ \\[8pt]
\hline
$\infty$ & $\lim_{J\to\infty} P_{J,1}((J!)^{1/J})$ & $1$ \\
\hline
\end{tabular}
\caption{Optimized lower bounds for $b=1$ and $\kappa=\infty$, rounded down to
three decimal places.}
\label{tab:exact-ratios}
\end{table}

Table~\ref{tab:exact-ratios} studies the regime $b=1$ and $J\to\infty$. In
this regime, each right vertex contributes one element to $\OPT$, but the
algorithm has $J$ tracks available. The optimized lower bound tends to $1$.

The regime $J=1$ and $b\to\infty$ is different. There is only one track, and the
underlying transversal matroid has larger right-side capacities. Hence each
right vertex can contribute $b$ elements to $\OPT$. Our guarantee is
element-wise: it asks that each element of $\OPT$ appears in the output with a
prescribed probability. Thus increasing the number of tracks and increasing the
right-side capacities affect different objects. The next calculation shows that
these two regimes lead to different limits.

\paragraph{The limit $b\to\infty$.}
\label{sec:largeb}

Fix $J$ and a parameter $\lambda>0$. For this parameter, the uncapacitated
routing ratio is $P_{J,b}(\lambda)$, where the threshold is $p=e^{-\lambda}$.
For every $z<J$,
\[
    \Pr\left(\Po(bz)\le bJ-1\right)\longrightarrow 1
    \qquad
    \text{as } b\to\infty,
\]
whereas for every $z>J$,
\[
    \Pr\left(\Po(bz)\le bJ-1\right)\longrightarrow 0
    \qquad
    \text{as } b\to\infty.
\]
The point $z=J$ does not affect the integral. Since the integrand in
Lemma~\ref{lem:local_prob} is bounded by $e^z$ on $[0,\lambda]$, dominated
convergence gives
\[
    \lim_{b\to\infty}P_{J,b}(\lambda)
    =
    e^{-\lambda}\int_0^\lambda e^z\mathbf{1}_{\{z<J\}}\,dz
    =
    \begin{cases}
    1-e^{-\lambda}, & 0<\lambda\le J,\\
    e^{J-\lambda}-e^{-\lambda}, & \lambda>J.
    \end{cases}
\]
The limiting objective is increasing on $(0,J]$ and decreasing on $(J,\infty)$.
Thus, after optimizing over $\lambda$,
\[
    \sup_{\lambda>0}\lim_{b\to\infty}P_{J,b}(\lambda)
    =
    1-e^{-J}.
\]
In particular, for $J=1$, the optimized limiting value is $1-1/e$. This is
different from the limit $1$ obtained when $b=1$ and $J\to\infty$.

\begin{remark}
\label{rem:berczi-connection}
B\'erczi, Livanos, Soto, and Verdugo~\cite{BercziLSV24} showed that
\textsc{Greedy-Improving} on a uniform matroid tends to $1-1/e$ as the rank
goes to infinity. The formula above recovers the same value in the regime
$J=1$ and $b\to\infty$:
\[
    \sup_{\lambda>0}\lim_{b\to\infty}P_{1,b}(\lambda)
    =
    1-\frac{1}{e}.
\]
This is the regime where the rank of the underlying matroid grows. The regime
$b=1$ and $J\to\infty$ is different: the set $\OPT$ has only one element per
right vertex, but the algorithm has more tracks. In that regime,
Theorem~\ref{thm:routing-asymptotic} gives $P_{J,1}(\lambda^*)\to 1$. Hence
$J$ and $b$ are not interchangeable, even though the local routing capacity at
each right vertex is $bJ$.
\end{remark}

\subsection{Finite \texorpdfstring{$\kappa$}{kappa} and \texorpdfstring{$q\to\infty$}{q to infty}}

We now turn to the case $\kappa<\infty$. The execution with parameter
$\kappa<\infty$ differs from the execution with parameter $\kappa=\infty$ only
through the global constraint $|\AUX|\le \kappa bq$. The next argument compares
these two executions under the same arrivals, weights, neighborhoods, and
tie-breaking rules.

\begin{definition}[Per-vertex load]
Fix $p=e^{-\lambda}$. For each $f\in R$, let $S_f\subseteq [bq]$ be the block
of labels assigned to $f$, so $|S_f|=b$, and let
\[
    N_f:=N_{S_f}[p,1).
\]
By Proposition~\ref{prop:poisson-label-general}, the variables $N_f$, $f\in R$,
are mutually independent and satisfy $N_f\sim\Po(b\lambda)$. Define
\[
    \widehat A_f:=\min\{bJ,N_f\},
    \qquad
    \widehat A:=\sum_{f\in R}\widehat A_f,
\]
and
\[
    \delta_{J,b}(\lambda):=\E[\widehat A_f]
    =
    \E[\min\{bJ,\Po(b\lambda)\}].
\]
In the execution with parameter $\kappa=\infty$, $\widehat A_f$ is the number
of elements routed through $f$, and $\widehat A$ is the final size of the
candidate pool $\AUX$.
\end{definition}

We will use that $\delta_{J,b}$ is continuous and strictly increasing. Indeed,
if $L=bJ$ and $\mu=b\lambda$, then
\[
    \E[\min\{L,\Po(\mu)\}]
    =
    \sum_{i=1}^{L}\Pr(\Po(\mu)\ge i).
\]
Therefore,
\[
    \frac{d}{d\lambda}\delta_{J,b}(\lambda)
    = b\,\Pr\bigl(\Po(b\lambda)\le bJ-1\bigr) > 0.
\]

\begin{theorem}
\label{thm:global_b_kappa}
Fix $J\ge 1$, $b\ge 1$, and $\kappa>0$. Let $\lambda>0$ satisfy
$\delta_{J,b}(\lambda)<b\kappa$, and set $p=e^{-\lambda}$. Then
\[
    \Prob{J}{\kappa}{b}{q}(\lambda)
    \ge
    P_{J,b}(\lambda)
    -
    \exp\left(
        -\frac{2q(b\kappa-\delta_{J,b}(\lambda))^2}{(bJ)^2}
    \right).
\]
Consequently,
\[
    \lim_{q\to\infty}\Prob{J}{\kappa}{b}{q}(\lambda)
    =
    P_{J,b}(\lambda).
\]
\end{theorem}

\begin{proof}
Use the coupling described above. In the execution with parameter
$\kappa=\infty$, the final size of the candidate pool $\AUX$ is
$\widehat A=\sum_{f\in R}\widehat A_f$. The variables $\widehat A_f$ are
mutually independent, lie in $[0,bJ]$, and satisfy
$\E[\widehat A]=q\delta_{J,b}(\lambda)$. Hoeffding's inequality gives
\[
    \Pr(\widehat A>\kappa bq)
    \le
    \exp\left(
        -\frac{2q(b\kappa-\delta_{J,b}(\lambda))^2}{(bJ)^2}
    \right).
\]
On the event $\widehat A\le \kappa bq$, the execution with parameter
$\kappa=\infty$ never violates the global bound $|\AUX|\le \lfloor\kappa bq
\rfloor$, because $|\AUX|$ is nondecreasing and its final value is an integer.
Thus the two coupled executions add the same elements to $\AUX$ on this event.

Fix an original element $v^*\in\OPT$. Under the coupling above, the capacitated execution ($\kappa<\infty$) agrees with the uncapacitated execution ($\kappa=\infty$) on the event $\{\widehat A\le \kappa bq\}$. Therefore,
\[
\Pr(v^*\in\ALG) \ge P_{J,b}(\lambda) - \Pr(\widehat A>\kappa bq),
\]
which gives the stated lower bound.

For the matching upper bound, note that the execution with $\kappa<\infty$ never adds an element that the execution with $\kappa=\infty$ rejects. Hence, by Observation~\ref{obs:union-vs-output},
\[
    \Prob{J}{\kappa}{b}{q}(\lambda) \le P_{J,b}(\lambda).
\]
Together, these inequalities imply the convergence claim.
\end{proof}

\begin{remark}
\label{rem:boundary}
The strict inequality $\delta_{J,b}(\lambda)<b\kappa$ in
Theorem~\ref{thm:global_b_kappa} is only used for the finite-$q$ error term. It
is not needed for the limiting optimization problem. Suppose that
$\delta_{J,b}(\lambda)=b\kappa$. For every $\lambda'<\lambda$, we have
$\delta_{J,b}(\lambda')<b\kappa$, and Theorem~\ref{thm:global_b_kappa} gives
\[
    \lim_{q\to\infty} \Prob{J}{\kappa}{b}{q}(\lambda')
    =
    P_{J,b}(\lambda').
\]
Since $P_{J,b}$ is continuous by Lemma~\ref{lem:local_prob}, these values
converge to $P_{J,b}(\lambda)$ as $\lambda'\uparrow\lambda$. Thus the boundary
value is approached by thresholds with strict slack, and allowing boundary
points does not change the limiting supremum.
\end{remark}

\begin{remark}
\label{rem:overloaded-threshold}
Fix $J$, $b$, and $\kappa$. Consider a parameter $\lambda$ such that $\delta_{J,b}(\lambda)>b\kappa$. In the limit $q\to\infty$, this choice cannot be optimal.

By the monotonicity of $\delta_{J,b}$, there is a unique $\lambda_\kappa\in(0,\lambda)$ such that
\[
    \delta_{J,b}(\lambda_\kappa)=b\kappa.
\]

Run the algorithm with threshold $p=e^{-\lambda}$. For $0\le z\le\lambda$, let $\widehat A_\lambda(z)$ be the number of elements that the execution with $\kappa=\infty$ adds to $\AUX$ during $[p,pe^z]$. Then $\E[\widehat A_\lambda(z)]=q\delta_{J,b}(z)$ and $\widehat A_\lambda(z)$ is a sum of $q$ independent random variables in $[0,bJ]$.

Fix $0<\varepsilon<\lambda-\lambda_\kappa$. Since $\delta_{J,b}(\lambda_\kappa+\varepsilon)>b\kappa$, Hoeffding's inequality implies
\[
    \Pr\bigl(\widehat A_\lambda(\lambda_\kappa+\varepsilon)\le \kappa bq\bigr)\to 0
    \qquad\text{as }q\to\infty.
\]
Thus, with probability tending to one, the execution with finite $\kappa$ fills
$\AUX$ by time $pe^{\lambda_\kappa+\varepsilon}$. For a fixed $v^*\in\OPT$, this
first yields
\[
    \limsup_{q\to\infty} \Prob{J}{\kappa}{b}{q}(\lambda)
    \le e^{-\lambda}\int_0^{\lambda_\kappa+\varepsilon}
    e^z\Pr(\Po(bz)\le bJ-1)\,dz.
\]
Letting $\varepsilon\downarrow0$ and using continuity of the integrand gives
\[
    \limsup_{q\to\infty} \Prob{J}{\kappa}{b}{q}(\lambda)
    \le e^{-\lambda}\int_0^{\lambda_\kappa} e^z\Pr(\Po(bz)\le bJ-1)\,dz.
\]
On the other hand,
\[
    P_{J,b}(\lambda_\kappa)= e^{-\lambda_\kappa}\int_0^{\lambda_\kappa} e^z\Pr(\Po(bz)\le bJ-1)\,dz.
\]
Since $\lambda_\kappa<\lambda$, the boundary value gives a strictly larger limit. Therefore, after taking $q\to\infty$, it is enough to maximize $P_{J,b}(\lambda)$ over the feasible region $\delta_{J,b}(\lambda)\le b\kappa$.
\end{remark}

\paragraph{The case $b=1$.}
\label{sec:b1}

We now specialize to $b=1$. Write
\[
    \delta_J(\lambda):=\delta_{J,1}(\lambda)
    =
    \E[\min\{J,N\}],
    \qquad N\sim\Po(\lambda).
\]
Equivalently,
\[
    \delta_J(\lambda)
    =
    J-e^{-\lambda}\sum_{k=0}^{J-1}(J-k)\frac{\lambda^k}{k!}.
\]

The preceding theorem and remarks reduce the limit $q\to\infty$ (with $J$, $\kappa$, and $b=1$ fixed) to the one-dimensional optimization problem
\[
    \alpha^\infty_{J,\kappa} := \sup\{P_{J,1}(\lambda): \lambda\ge 0,\ \delta_J(\lambda)\le \kappa\}.
\]

We now solve this one-dimensional optimization problem. Since
\[
    P'_{J,1}(\lambda)
    =
    e^{-\lambda}\left(1-\frac{\lambda^J}{J!}\right),
\]
the function $P_{J,1}$ is maximized over $\lambda\ge 0$ at
\[
    \lambda_J^{\mathrm{loc}}=(J!)^{1/J}.
\]
Also, $\delta_J$ is continuous and strictly increasing, because
\[
    \delta'_J(\lambda)=\Pr(\Po(\lambda)\le J-1)>0.
\]
Therefore, if $\delta_J(\lambda_J^{\mathrm{loc}})\le \kappa$, the capacity
constraint does not bind and
\[
    \alpha^\infty_{J,\kappa} = P_{J,1}(\lambda_J^{\mathrm{loc}}).
\]
Otherwise the capacity constraint binds, and the optimizer is the unique value
$\lambda^*_{J,\kappa}\in(0,\lambda_J^{\mathrm{loc}})$ satisfying
\[
    \delta_J(\lambda^*_{J,\kappa})=\kappa.
\]
Equivalently, in both cases,
\[
    \lambda^*_{J,\kappa}
    =
    \min\left\{
        \lambda_J^{\mathrm{loc}},
        \sup\{\lambda\ge 0:\delta_J(\lambda)\le\kappa\}
    \right\},
    \qquad
    \alpha^\infty_{J,\kappa}
    =
    P_{J,1}(\lambda^*_{J,\kappa}).
\]

\newgeometry{margin=1.5cm}
\begin{table}[ht]
\centering
\begingroup
\setlength{\tabcolsep}{3pt}
\begin{tabular}{c|ccccccccccc|c}
\toprule
$\kappa$ & $J=1$ & $J=2$ & $J=3$ & $J=4$ & $J=5$ & $J=6$ & $J=7$ & $J=8$ & $J=9$ & $J=10$ & $J=100$ & $J \to \infty$ \\
\midrule
0.0 & 0.000 & 0.000 & 0.000 & 0.000 & 0.000 & 0.000 & 0.000 & 0.000 & 0.000 & 0.000 & 0.000 & 0.000 \\
0.2 & 0.179 & 0.181 & 0.181 & 0.181 & 0.181 & 0.181 & 0.181 & 0.181 & 0.181 & 0.181 & 0.181 & 0.181 \\
0.4 & 0.306 & 0.327 & 0.329 & 0.330 & 0.330 & 0.330 & 0.330 & 0.330 & 0.330 & 0.330 & 0.330 & 0.330 \\
0.6 & 0.367 & 0.442 & 0.450 & 0.451 & 0.451 & 0.451 & 0.451 & 0.451 & 0.451 & 0.451 & 0.451 & 0.451 \\
0.8 & \textbf{0.368} & 0.524 & 0.546 & 0.550 & 0.551 & 0.551 & 0.551 & 0.551 & 0.551 & 0.551 & 0.551 & 0.551 \\
1.0 & \textbf{0.368} & 0.573 & 0.621 & 0.630 & 0.632 & 0.632 & 0.632 & 0.632 & 0.632 & 0.632 & 0.632 & 0.632 \\
1.2 & \textbf{0.368} & \textbf{0.587} & 0.675 & 0.694 & 0.698 & 0.699 & 0.699 & 0.699 & 0.699 & 0.699 & 0.699 & 0.699 \\
1.4 & \textbf{0.368} & \textbf{0.587} & 0.710 & 0.743 & 0.751 & 0.753 & 0.753 & 0.753 & 0.753 & 0.753 & 0.753 & 0.753 \\
1.6 & \textbf{0.368} & \textbf{0.587} & 0.725 & 0.779 & 0.794 & 0.797 & 0.798 & 0.798 & 0.798 & 0.798 & 0.798 & 0.798 \\
1.8 & \textbf{0.368} & \textbf{0.587} & \textbf{0.726} & 0.803 & 0.826 & 0.833 & 0.834 & 0.835 & 0.835 & 0.835 & 0.835 & 0.835 \\
2.0 & \textbf{0.368} & \textbf{0.587} & \textbf{0.726} & 0.815 & 0.850 & 0.861 & 0.864 & 0.864 & 0.865 & 0.865 & 0.865 & 0.865 \\
2.2 & \textbf{0.368} & \textbf{0.587} & \textbf{0.726} & \textbf{0.817} & 0.867 & 0.883 & 0.887 & 0.889 & 0.889 & 0.889 & 0.889 & 0.889 \\
2.4 & \textbf{0.368} & \textbf{0.587} & \textbf{0.726} & \textbf{0.817} & 0.875 & 0.899 & 0.906 & 0.909 & 0.909 & 0.909 & 0.909 & 0.909 \\
2.6 & \textbf{0.368} & \textbf{0.587} & \textbf{0.726} & \textbf{0.817} & \textbf{0.877} & 0.910 & 0.921 & 0.924 & 0.925 & 0.926 & 0.926 & 0.926 \\
2.8 & \textbf{0.368} & \textbf{0.587} & \textbf{0.726} & \textbf{0.817} & \textbf{0.877} & 0.915 & 0.932 & 0.937 & 0.939 & 0.939 & 0.939 & 0.939 \\
3.0 & \textbf{0.368} & \textbf{0.587} & \textbf{0.726} & \textbf{0.817} & \textbf{0.877} & \textbf{0.917} & 0.939 & 0.947 & 0.949 & 0.950 & 0.950 & 0.950 \\
3.2 & \textbf{0.368} & \textbf{0.587} & \textbf{0.726} & \textbf{0.817} & \textbf{0.877} & \textbf{0.917} & 0.943 & 0.954 & 0.958 & 0.959 & 0.959 & 0.959 \\
3.4 & \textbf{0.368} & \textbf{0.587} & \textbf{0.726} & \textbf{0.817} & \textbf{0.877} & \textbf{0.917} & \textbf{0.944} & 0.959 & 0.964 & 0.966 & 0.967 & 0.967 \\
3.6 & \textbf{0.368} & \textbf{0.587} & \textbf{0.726} & \textbf{0.817} & \textbf{0.877} & \textbf{0.917} & \textbf{0.944} & 0.961 & 0.969 & 0.971 & 0.973 & 0.973 \\
3.8 & \textbf{0.368} & \textbf{0.587} & \textbf{0.726} & \textbf{0.817} & \textbf{0.877} & \textbf{0.917} & \textbf{0.944} & \textbf{0.962} & 0.972 & 0.976 & 0.978 & 0.978 \\
4.0 & \textbf{0.368} & \textbf{0.587} & \textbf{0.726} & \textbf{0.817} & \textbf{0.877} & \textbf{0.917} & \textbf{0.944} & \textbf{0.962} & 0.974 & 0.979 & 0.982 & 0.982 \\
5.0 & \textbf{0.368} & \textbf{0.587} & \textbf{0.726} & \textbf{0.817} & \textbf{0.877} & \textbf{0.917} & \textbf{0.944} & \textbf{0.962} & \textbf{0.974} & \textbf{0.982} & 0.993 & 0.993 \\
6.0 & \textbf{0.368} & \textbf{0.587} & \textbf{0.726} & \textbf{0.817} & \textbf{0.877} & \textbf{0.917} & \textbf{0.944} & \textbf{0.962} & \textbf{0.974} & \textbf{0.982} & 0.998 & 0.998 \\
7.0 & \textbf{0.368} & \textbf{0.587} & \textbf{0.726} & \textbf{0.817} & \textbf{0.877} & \textbf{0.917} & \textbf{0.944} & \textbf{0.962} & \textbf{0.974} & \textbf{0.982} & 0.999 & 0.999 \\
8.0 & \textbf{0.368} & \textbf{0.587} & \textbf{0.726} & \textbf{0.817} & \textbf{0.877} & \textbf{0.917} & \textbf{0.944} & \textbf{0.962} & \textbf{0.974} & \textbf{0.982} & 1.000 & 1.000 \\
9.0 & \textbf{0.368} & \textbf{0.587} & \textbf{0.726} & \textbf{0.817} & \textbf{0.877} & \textbf{0.917} & \textbf{0.944} & \textbf{0.962} & \textbf{0.974} & \textbf{0.982} & 1.000 & 1.000 \\
10.0 & \textbf{0.368} & \textbf{0.587} & \textbf{0.726} & \textbf{0.817} & \textbf{0.877} & \textbf{0.917} & \textbf{0.944} & \textbf{0.962} & \textbf{0.974} & \textbf{0.982} & 1.000 & 1.000 \\
100.0 & \textbf{0.368} & \textbf{0.587} & \textbf{0.726} & \textbf{0.817} & \textbf{0.877} & \textbf{0.917} & \textbf{0.944} & \textbf{0.962} & \textbf{0.974} & \textbf{0.982} & \textbf{1.000} & 1.000 \\
\midrule
$\kappa = J$ & \textbf{0.368} & \textbf{0.587} & \textbf{0.726} & \textbf{0.817} & \textbf{0.877} & \textbf{0.917} & \textbf{0.944} & \textbf{0.962} & \textbf{0.974} & \textbf{0.982} & \textbf{1.000} & \textbf{1.000} \\
\bottomrule
\end{tabular}
\endgroup
\caption{Asymptotic optimized probability-competitive ratios for $b=1$:
$\alpha^\infty_{J,\kappa}=P_{J,1}(\lambda^*_{J,\kappa})$. Here $\lambda^*_{J,\kappa}$ maximizes $P_{J,1}(\lambda)$ subject to $\delta_J(\lambda)\le \kappa$. Bold entries indicate that the capacity constraint is not active (equivalently, $\delta_J(\lambda_J^{\mathrm{loc}})\le\kappa$). Values are rounded to three decimals.}
\label{tab:prob_kappa}
\end{table}
\restoregeometry

\begin{remark}
\label{rem:Jinfty}
Fix $\kappa<\infty$ and let $J\to\infty$. For every fixed (bounded) $\lambda$, we have $\delta_J(\lambda)\to\lambda$ and $P_{J,1}(\lambda)\to 1-e^{-\lambda}$. Since $\lambda_J^{\mathrm{loc}}=(J!)^{1/J}\to\infty$, the capacity constraint is active for all large $J$, and $\lambda^*_{J,\kappa}\to\kappa$. Therefore,
\[
    \lim_{J\to\infty}\alpha^\infty_{J,\kappa} = 1-e^{-\kappa}.
\]
This differs from the diagonal choice $\kappa=J$, where the constraint is not active and the optimized value $P_{J,1}(\lambda_J^{\mathrm{loc}})$ tends to $1$ (Theorem~\ref{thm:routing-asymptotic}).
\end{remark}

Figure~\ref{fig:kappa-limit} illustrates the optimized limiting value
$\alpha^\infty_{J,\kappa}$ for $b=1$ and $q\to\infty$.

\begin{figure}[t]
\centering
\includegraphics[width=.84\textwidth]{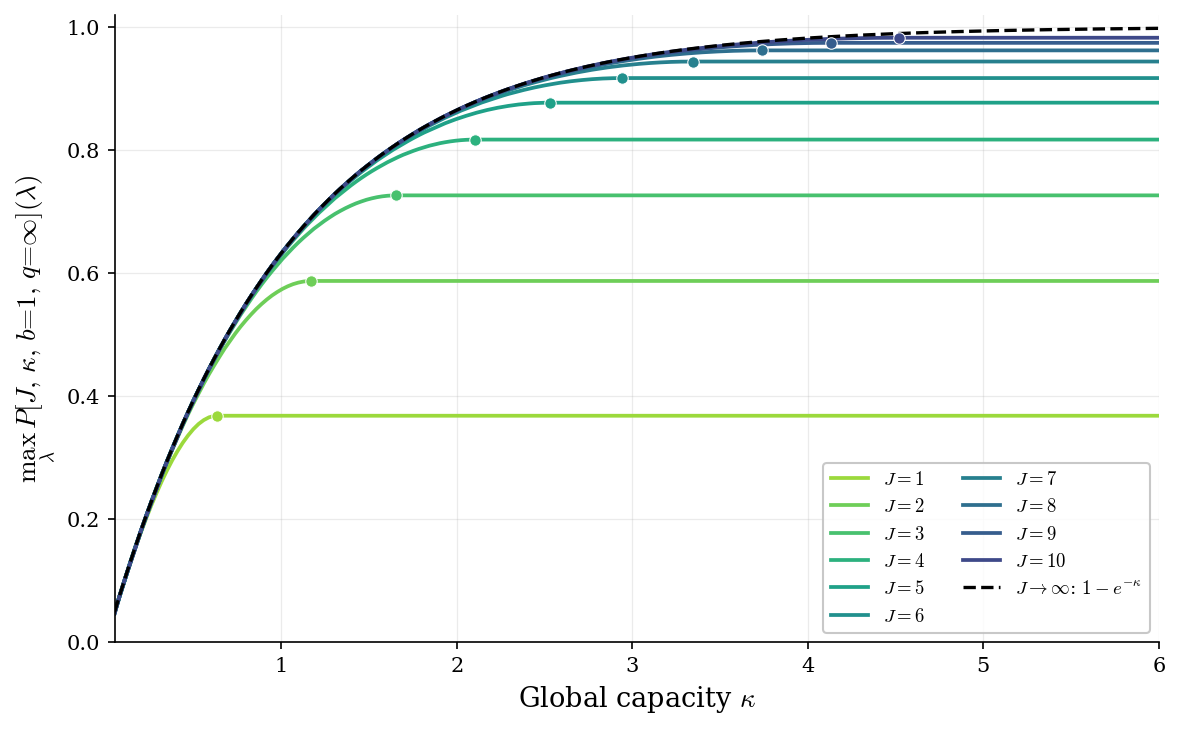}
\caption{Optimized limiting probability-competitive ratio for the
single-threshold routing algorithm with $b=1$ and $q\to\infty$. For each $J$,
the curve is
$\alpha^\infty_{J,\kappa}=\sup\{P_{J,1}(\lambda):\delta_J(\lambda)\le\kappa\}$.
The marker on each curve is the first value
$\kappa=\delta_J(\lambda_J^{\mathrm{loc}})$ at which the capacity constraint is
not active, where $\lambda_J^{\mathrm{loc}}=(J!)^{1/J}$. The dashed curve is
the limit $1-e^{-\kappa}$ as $J\to\infty$ for fixed $\kappa$.}
\label{fig:kappa-limit}
\end{figure}

\paragraph{The limit $b\to\infty$.}

Fix $J$ and $\kappa$. We first take $q\to\infty$, as in this subsection, and then let $b\to\infty$ with $\lambda$ fixed. For each fixed $\lambda$,
\[
    \frac{1}{b}\delta_{J,b}(\lambda)=\E\left[\min\left\{J,\frac{\Po(b\lambda)}{b}\right\}\right].
\]
Since $\Po(b\lambda)/b\to\lambda$ in probability and $x\mapsto \min\{J,x\}$ is bounded and continuous on $\mathbb R_+$,
\[
    \frac{1}{b}\delta_{J,b}(\lambda)\longrightarrow \min\{J,\lambda\}.
\]
Hence the constraint $\delta_{J,b}(\lambda)\le b\kappa$ converges to $\min\{J,\lambda\}\le\kappa$.

The pointwise limit of $P_{J,b}(\lambda)$ was computed in Section~\ref{sec:largeb}: $P_{J,b}(\lambda)\to L_J(\lambda)$, where
\[
    L_J(\lambda)=
    \begin{cases}
    1-e^{-\lambda}, & 0\le \lambda\le J,\\
    e^{J-\lambda}-e^{-\lambda}, & \lambda>J.
    \end{cases}
\]
The corresponding limiting optimization problem is
\[
    \sup\{L_J(\lambda): \lambda\ge 0,\ \min\{J,\lambda\}\le\kappa\}=1-e^{-\min\{\kappa,J\}}.
\]
\subsection{Finite \texorpdfstring{$\kappa$}{kappa} and finite \texorpdfstring{$q$}{q}}

Theorem~\ref{thm:global_b_kappa} already gives a finite-$q$ bound for every $\lambda$ with $\delta_{J,b}(\lambda)<b\kappa$, but the exponent depends on $J$ through the range $bJ$ of each variable $\widehat A_f$. We next record a simpler bound that is uniform in $J$, under the stronger assumption $0<\lambda<\kappa$.

\begin{proposition}
\label{prop:finite-q}
Fix $J\ge 1$, $b\ge 1$, $\kappa>0$, and $0<\lambda<\kappa$. For every $\varepsilon\in(0,1)$, if
\[
    q \ge q_\varepsilon := \frac{\ln(1/\varepsilon)}{b\,D(\kappa\,\|\,\lambda)},
\]
then
\[
    \Prob{J}{\kappa}{b}{q}(\lambda) \ge P_{J,b}(\lambda)-\varepsilon,
\]
where
\[
    D(\kappa\|\lambda) := \kappa\ln(\kappa/\lambda)-\kappa+\lambda > 0.
\]
In particular, $q_\varepsilon$ does not depend on $J$.
\end{proposition}

\begin{proof}
For each $f\in R$, we have $\widehat A_f\le N_f$, where the random variables $N_f$ are independent and satisfy $N_f\sim\Po(b\lambda)$. Let
\[
    N:=\sum_{f\in R}N_f.
\]
Then $N\sim\Po(qb\lambda)$ and $\widehat A\le N$.

Set $\mu=qb\lambda$. By the Cram\'er--Chernoff bound for Poisson variables (Lemma~\ref{lem:poisson-chernoff}(a)),
\[
    \Pr(N>\kappa bq)
    \le \Pr(N\ge \kappa bq)
    \le \exp\bigl(-D(\kappa bq\|\mu)\bigr).
\]
Since $D(\kappa bq\|\mu)=D(\kappa bq\|qb\lambda)=qb\,D(\kappa\|\lambda)$, we obtain
\[
    \Pr(\widehat A>\kappa bq) \le \Pr(N>\kappa bq) \le \exp\bigl(-qb\,D(\kappa\|\lambda)\bigr).
\]
If $q\ge q_\varepsilon$, this probability is at most $\varepsilon$.

Use the same coupling as in the proof of
Theorem~\ref{thm:global_b_kappa}. In the execution with parameter
$\kappa=\infty$, the final size of the candidate pool $\AUX$ is $\widehat A$.
On the event $\widehat A\le \kappa bq$, this integer final size is at most
$\lfloor \kappa bq\rfloor$. Since $|\AUX|$ is nondecreasing during the
execution, the uncapacitated execution never violates the global capacity bound
on this event. Hence the executions with parameters $\kappa<\infty$ and
$\kappa=\infty$ add the same elements to $\AUX$ on this event.

Therefore, for every fixed original element $v^*\in\OPT$,
\[
    \Prob{J}{\kappa}{b}{q}(\lambda)
    \ge
    P_{J,b}(\lambda)-\Pr(\widehat A>\kappa bq)
    \ge
    P_{J,b}(\lambda)-\varepsilon .\qedhere
\]
\end{proof}

Figure~\ref{fig:finite-q-bound} illustrates the finite-$q$ bound of
Proposition~\ref{prop:finite-q} for one fixed choice of parameters.

\begin{figure}[t]
\centering
\includegraphics[width=.78\textwidth]{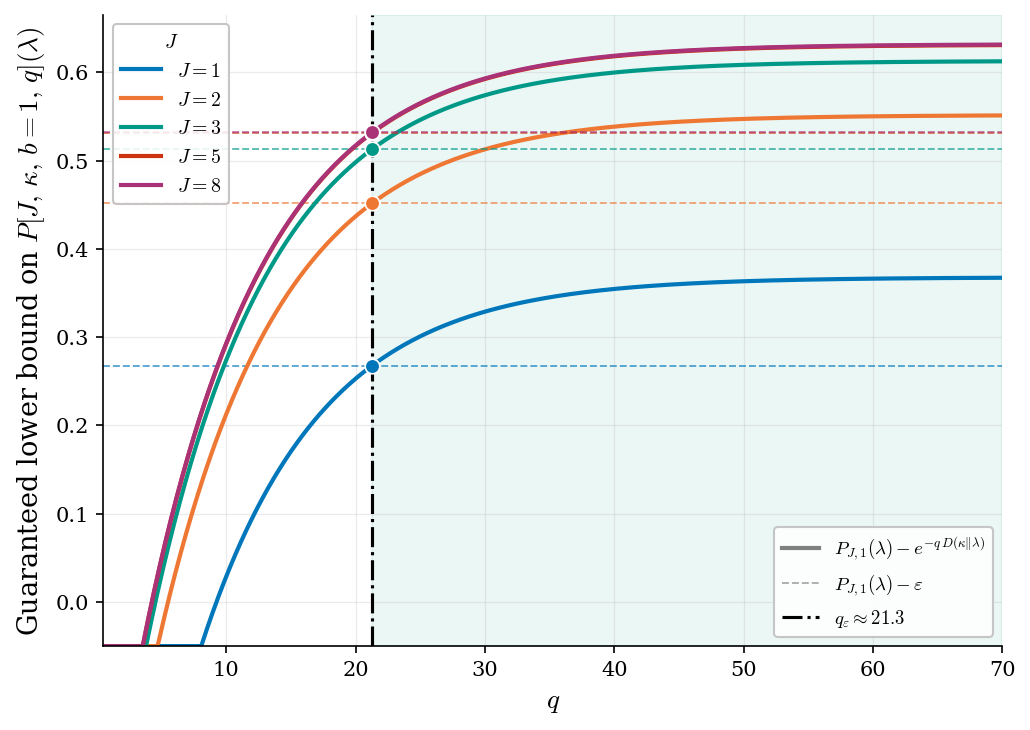}
\caption{Finite-$q$ lower bound from Proposition~\ref{prop:finite-q} for
$b=1$, $\kappa=1.5$, $\lambda=1$, and $\varepsilon=0.1$. For each displayed
value of $J$, the solid curve is
$P_{J,1}(1)-\exp(-qD(1.5\,\|\,1))$. The dashed horizontal line is
$P_{J,1}(1)-0.1$. The vertical dash-dotted line marks
$q_\varepsilon=\ln(10)/D(1.5\,\|\,1)\approx 21.28$. Thus, for
$q\ge q_\varepsilon$, the displayed bound is within $0.1$ of $P_{J,1}(1)$,
uniformly over the displayed values of $J$.}
\label{fig:finite-q-bound}
\end{figure}

\paragraph{The case $b=1$.}

For $b=1$, Proposition~\ref{prop:finite-q} gives a finite-$q$ guarantee for
every parameter $0<\lambda<\kappa$. This condition is stronger than the
condition $\delta_J(\lambda)\le\kappa$ used in the limiting optimization of the
previous subsection.

For example, take $J=2$ and $\kappa=1$. Proposition~\ref{prop:finite-q} applies
to every fixed $\lambda<1$. As $\lambda\uparrow 1$,
\[
    P_{2,1}(\lambda)
    \longrightarrow
    P_{2,1}(1)
    =
    e^{-1}\left(1+\frac{1}{2}\right)
    =
    \frac{3}{2e}
    \approx 0.552 .
\]
Thus the proposition gives guarantees arbitrarily close to $3/(2e)$ by taking
$\lambda<1$ sufficiently close to $1$, with $q_\varepsilon$ depending on this
choice of $\lambda$.

The optimizer in the limiting problem of the previous subsection is different.
It is the solution of
\[
    \delta_2(\lambda)
    =
    2-(2+\lambda)e^{-\lambda}
    =
    1,
\]
namely $\lambda\approx 1.146$. This parameter is larger than $\kappa=1$, so it
is not covered by Proposition~\ref{prop:finite-q}. To get a finite-$q$ bound
approaching this limiting value, one must use
Theorem~\ref{thm:global_b_kappa} with parameters $\lambda'<1.146$ satisfying
$\delta_2(\lambda')<1$, and then let $\lambda'$ approach the root. The rate then
depends on the gap $1-\delta_2(\lambda')$.

This completes the analysis of routing algorithms for capacitated transversal
matroids. The next section extends the $J$-MSP framework to $k$-column-sparse
and laminar matroids, where the routing structure and feasibility certificates
take a different form.

\section[Beyond transversal matroids: k-column-sparse and laminar matroids]{Beyond transversal matroids: \texorpdfstring{$k$}{k}-column-sparse and \\ laminar matroids}
\label{sec:examples}

The transversal-matroid setting admits a canonical routing rule: each improving
arrival is routed to the offline vertex that certifies its membership in the
current optimum, and the fate of a fixed $e^*\in\OPT$ can be followed using a
single forbidden label. In the terminology of Soto, Turkieltaub, and
Verdugo~\cite{SotoTV18}, this corresponds to forbidden sets of size $1$.

This section illustrates how the $(J,\kappa)$-MSP framework extends beyond
transversal matroids by presenting two representative classes where the
certificate structure is no longer induced by routing to offline vertices.

For $k$-column-sparse matroids, we keep the multi-track paradigm: the algorithm
maintains $J$ independent sets and, upon accepting an element into a track,
records a row index of the representing matrix from its support. For a fixed
$e^*\in\OPT$, the analysis focuses on the conflicts created on the (at most $k$)
rows where the column of $e^*$ is nonzero, yielding at most $k$ forbidden
labels.

For laminar matroids, we adopt a different approach: the algorithm maintains a
single candidate pool and enforces feasibility directly in $\calM^{(J)}$, never
routing elements online to specific tracks.

\subsection{Multi-track algorithms for \texorpdfstring{$k$}{k}-column-sparse matroids}
\label{sec:ksparse}

\subsubsection{The multi-track algorithm}

Let $\calM$ be a $k$-column-sparse matroid, as defined in
Section~\ref{sec:prelim}. Thus $\calM$ is represented by a matrix $M$ over a
field $\FF$, with rows indexed by $[m]$ and columns indexed by the ground set
$N$, and each column has at most $k$ nonzero entries. For $e\in N$, write
\[
    \supp(e):=\{a\in[m]:M_{a,e}\neq 0\}.
\]
We refer to $\supp(e)$ as the (row) \emph{support} of the column corresponding to $e$.

We use the following standard consequence of this representation, as in
Soto~\cite{Soto13}. Let $X\subseteq N$ be independent and consider the bipartite
graph with left side $X$ and right side $[m]$, where $e\in X$ is adjacent to the
rows in $\supp(e)$. For every $Y\subseteq X$, its neighborhood has size at least
$|Y|$; otherwise, the columns of $Y$ would be supported on fewer than $|Y|$ rows
and would be linearly dependent. Hence Hall's condition holds, so this graph has
a matching that covers $X$.

We select one such matching using a deterministic tie-breaking rule applied to
the bipartite graph induced by $X$ (equivalently, to the columns revealed in
$X$). This defines an injective map $\pi_X:X\to[m]$ such that
$\pi_X(e)\in\supp(e)$ for every $e\in X$. We call $\pi_X(e)$ the
\emph{canonical row} of $e$ with respect to $X$.

The algorithm maintains $J$ independent sets $F_1,\dots,F_J$. For each track
$i\in[J]$, it also maintains a set $P_i\subseteq[m]$ of row indices already used
in track $i$. When the algorithm inserts an element $e$ into track $i$ at time
$t$, it stores the canonical row $\pi_{\OPT_t}(e)\in\supp(e)$ in $P_i$. We call
this row index the \emph{stored row} of $e$ in track $i$.

When an improving element $e$ arrives at time $t$, the algorithm computes the
current optimum $\OPT_t$ and its canonical row map $\pi_{\OPT_t}$. It then scans
tracks in increasing order. Track $i$ is \emph{safe} for $e$ if none of the row
indices in $\supp(e)$ has been used before in that track, i.e., if
\[
    \supp(e)\cap P_i=\emptyset.
\]
If a safe track exists, the algorithm inserts $e$ into the first such track $i$
and stores $\pi_{\OPT_t}(e)$ in $P_i$.

Thus the algorithm uses the full support $\supp(e)$ to test safety, but stores
only a single row index, namely $\pi_{\OPT_t}(e)$. It also enforces the global
capacity constraint: if $\kappa<\infty$ and accepting $e$ would make the total
number of accepted elements exceed $\lfloor\kappa r\rfloor$, then $e$ is
rejected even if a safe track exists.
Algorithm~\ref{alg:ksparse} gives the pseudocode.

\begin{algorithm}[ht]
\caption{Multi-track algorithm for $k$-column-sparse matroids}
\label{alg:ksparse}
\DontPrintSemicolon
\KwData{A $k$-column-sparse matroid $\calM$ with matrix representation $M$, a
global capacity parameter $\kappa\in(0,\infty]$, an integer $J\ge 1$, a
threshold $p\in(0,1)$, and access to the arrival time $t_e$ and column $M_e$ of
each element $e$.}
\KwResult{A feasible output for the $(J,\kappa)$-MSP.}

Initialize $F_i\leftarrow\emptyset$ and $P_i\leftarrow\emptyset$ for all
$i\in[J]$, and initialize $\AUX\leftarrow\emptyset$ and $A\leftarrow 0$\;
Set $K\leftarrow\infty$ if $\kappa=\infty$, and
$K\leftarrow\lfloor\kappa r\rfloor$ otherwise\;
Reject all arrivals before time $p$\;

\For{each arrival time $t\in[p,1)$ in increasing order}{
    let $e_t$ be the arriving element at time $t$\;
    compute the canonical maximum-weight independent set $\OPT_t$ and the
    canonical row map $\pi_{\OPT_t}$\;

    \If{$e_t\in\OPT_t$ \textbf{and} $A<K$}{
        \If{there exists $i\in[J]$ such that $\supp(e_t)\cap P_i=\emptyset$}{
            let $i^*$ be the smallest index $i\in[J]$ such that $\supp(e_t)\cap P_i=\emptyset$\;
            add $e_t$ to $F_{i^*}$ and to $\AUX$; set
            $P_{i^*}\leftarrow P_{i^*}\cup\{\pi_{\OPT_t}(e_t)\}$; set
            $A\leftarrow A+1$\;
        }
    }
}

\Return{$\OPT(\AUX)$}\;
\end{algorithm}

The next lemma shows that each track maintained by Algorithm~\ref{alg:ksparse}
remains independent, and that the auxiliary set respects the global capacity
constraint.

\begin{lemma}
\label{prop:ksparse-feasibility}
Throughout the execution of Algorithm~\ref{alg:ksparse}, each set $F_i$ is
independent in $\calM$. Hence $\AUX$ is independent in $\calM^{(J)}$. Moreover,
if $\kappa<\infty$, then $|\AUX|\le \lfloor\kappa r\rfloor$.
\end{lemma}

\begin{proof}
Fix a track $i\in[J]$. Let $e_1,\ldots,e_s$ be the elements inserted into
$F_i$, ordered so that $t_{e_1}<\cdots<t_{e_s}$. For each $j\in[s]$, define the
stored row index
\[
    q_j:=\pi_{\OPT_{t_{e_j}}}(e_j)\in[m].
\]

Before inserting $e_j$, the set $P_i$ contains $q_1,\ldots,q_{j-1}$. The
algorithm inserts $e_j$ into track $i$ only if $\supp(e_j)\cap P_i=\emptyset$.
Since $q_j\in\supp(e_j)$, we have $q_j\notin\{q_1,\ldots,q_{j-1}\}$. Hence the
row indices $q_1,\ldots,q_s$ are distinct.

Let $Q=(q_1,\ldots,q_s)$ and $E_i=(e_1,\ldots,e_s)$. Use the order of $Q$ for
the rows and the order of $E_i$ for the columns, and consider the restricted
matrix $M[Q,E_i]$. Its diagonal entries are nonzero because
$q_j\in\supp(e_j)$ for every $j$. Its entries above the diagonal are zero:
if $\ell<j$, then $q_\ell\in P_i$ when the algorithm inserts $e_j$, and the
condition $\supp(e_j)\cap P_i=\emptyset$ gives $q_\ell\notin\supp(e_j)$.
Thus $M_{q_\ell,e_j}=0$.

Thus $M[Q,E_i]$ is lower triangular with nonzero diagonal, so its columns are
linearly independent. Since these columns are restrictions of the columns
$M_{e_1},\ldots,M_{e_s}$ of $M$, the full columns $M_{e_1},\ldots,M_{e_s}$ are
also linearly independent. Therefore $F_i$ is independent in $\calM$.

Since every $F_i$ is independent in $\calM$, the pool
$\AUX=F_1\cup\cdots\cup F_J$ is independent in $\calM^{(J)}$. The counter $A$
and the condition $A<K$ ensure $|\AUX|\le\lfloor\kappa r\rfloor$ when
$\kappa<\infty$.
\end{proof}

\subsubsection{Labeling scheme and row conflicts}

We next define a labeling scheme for the analysis of
Algorithm~\ref{alg:ksparse}. Fix an element $e^*\in\OPT$. The goal is to control
how many improving arrivals before $e^*$ store a row index in $\supp(e^*)$.

We use the full-rank convention from Section~\ref{sec:prelim}. Starting from a
$k$-column-sparse representation with rows indexed by $[m]$, we add dummy unit
columns of the $m\times m$ identity matrix. These dummy columns are part of the
augmented instance used in the analysis; no row preprocessing is performed.
The resulting augmented matroid has rank $m$. From this point on in the
analysis, we write $r=m$ for this rank. Moreover, for every $s\ge \tau_0$ the
set of revealed columns $N_s$ has full
rank $r$, so $\OPT_s$ is a basis of size $r$. Accordingly, we index rows by
$[r]$.

The dummy columns are $1$-sparse, so the augmented instance is still
$k$-column-sparse. As in the full-rank convention, dummy elements are processed
as ordinary elements in the augmented execution. They may be improving, may be
accepted, may be inserted into $\AUX$, may update the maintained sets, and may
consume global capacity when $\kappa<\infty$. All guarantees are stated only for
original elements of $\OPT$.

\paragraph{Labeling scheme for $k$-column-sparse matroids.}
Let $t_{e^*}$ be the arrival time of $e^*$, and write
\[
    \supp(e^*)=\{a_1,\ldots,a_h\}.
\]
Since each column has at most $k$ nonzero entries, we have $h\le k$ (and also
$h\le r$).

For each improving time $s\ge \tau_0$, the map
$\pi_{\OPT_s}:\OPT_s\to[r]$ is injective. Since $\OPT_s$ has size $r$, this map
is a bijection. Hence, for each row $a\in[r]$, there is a unique element of
$\OPT_s$ whose canonical row is $a$.

We assign labels as follows. For each $j\in[h]$, assign label $j$ to the unique
element of $\OPT_s$ whose canonical row is $a_j$. Then assign labels
$h+1,\ldots,r$ to the remaining elements of $\OPT_s$ in a fixed deterministic
order independent of the internal arrival order inside $N_s$.

This rule gives a bijection $\ell_s:\OPT_s\to[r]$ at every improving time
$s\ge \tau_0$. The rule may depend on $\OPT_s$ and on $\supp(e^*)$, but not
on the internal arrival order inside $N_s$. Hence it is a valid labeling scheme
in the sense of Section~\ref{par:labeling-schemes}.

Now consider an improving arrival at time $s<t_{e^*}$. If the row stored by
Algorithm~\ref{alg:ksparse} for this arrival lies in $\supp(e^*)$, then this
stored row is some $a_j$. By the labeling rule, the arriving element has label
$j$. Therefore, the number of improving arrivals before $t_{e^*}$ whose stored
row lies in $\supp(e^*)$ is at most $N_{[h]}[p,t_{e^*})$.

By Proposition~\ref{prop:poisson-label-general}, for every
$\tau_0\le s<t\le 1$,
\[
    N_{[h]}[s,t)\sim \Po\bigl(h\ln(t/s)\bigr),
\]
with independent increments on disjoint subintervals of $[\tau_0,1]$.
\subsubsection{Uncapacitated probability-competitive ratio}

We first analyze the case $\kappa=\infty$. We do not optimize the guarantee as in
the previous sections. We give a simple bound showing that the failure
probability decays exponentially with $J/k$. In particular, taking $J=ck$ for a
constant $c>0$, the ratio converges to $1$ at an exponential rate in $c$.

\begin{theorem}
\label{thm:ksparse}
For every $k$-column-sparse matroid $\calM$ and every $J \ge 1$,
Algorithm~\ref{alg:ksparse} with capacity parameter $\kappa=\infty$ and
threshold $p = e^{-J/(ke)}$ achieves a probability-competitive ratio of at
least $1 - O(e^{-J/(ke)})$.
\end{theorem}

\begin{proof}

Fix $e^*\in\OPT$, and let $t_{e^*}$ be its arrival time. We first condition on
$t_{e^*}=t\in[p,1)$.

Let $B$ be the number of \emph{accepted} elements arriving before $t$ whose
stored row index lies in $\supp(e^*)$. Each such accepted element $e$ arrives at
some time $s\in[p,t)$, is improving at time $s$, and stores a row index
$\pi_{\OPT_s}(e)\in\supp(e^*)=\{a_1,\ldots,a_h\}$. By the labeling rule at time
$s$, this implies that $e$ has a label in $[h]$. Hence $B\le N_{[h]}[p,t)$. By
Proposition~\ref{prop:poisson-label-general},
$N_{[h]}[p,t)\sim \Po(h\ln(t/p))$. Since $h\le k$, this random variable is
stochastically dominated by $\Po(k\ln(t/p))$.

If $B<J$, then not all tracks can contain a stored row index from $\supp(e^*)$.
Indeed, if $P_i\cap\supp(e^*)\neq\emptyset$ for every $i\in[J]$, then each
track contributes at least one accepted element counted in $B$, so $B\ge J$.
Consequently, there exists a track $i$ with
$P_i\cap\supp(e^*)=\emptyset$. This track is safe for $e^*$, so the algorithm
accepts $e^*$. Since $e^*\in\OPT$, Observation~\ref{obs:union-vs-output}
implies that $e^*\in\ALG$.

Therefore, conditional on $t_{e^*}=t$,
\[
    \Pr(e^*\in\ALG\mid t_{e^*}=t)
    \ge
    \Pr\bigl(\Po(k\ln(t/p))<J\bigr).
\]
Set $\lambda:=k\ln(t/p)$. Since $t\le 1$ and $p=e^{-J/(ke)}$, we have
$\lambda\le k\ln(1/p)=J/e$. By the Poisson Chernoff bound in
Lemma~\ref{lem:poisson-chernoff}(b),
\[
    \Pr(\Po(\lambda)\ge J)
    \le
    \left(\frac{e\lambda}{J}\right)^J e^{-\lambda}.
\]
The right-hand side is increasing for $\lambda<J$. Hence, on the interval
$[0,J/e]$, it is maximized at $\lambda=J/e$, where it equals $e^{-J/e}$.
Therefore,
\[
    \Pr(\Po(\lambda)\ge J)
    \le e^{-J/e}.
\]
Thus, for every $t\in[p,1)$,
\[
    \Pr(e^*\in\ALG\mid t_{e^*}=t)\ge 1-e^{-J/e}.
\]
Unconditioning over $t_{e^*}$ (uniform on $[0,1]$) and recalling that the
algorithm rejects every arrival before $p$, we obtain
\[
    \Pr(e^*\notin\ALG)
    =
    \int_0^1 \Pr(e^*\notin\ALG\mid t_{e^*}=t)\,dt
    \le
    p + \int_p^1 e^{-J/e}\,dt
    =
    p+(1-p)e^{-J/e}.
\]
Substituting $p=e^{-J/(ke)}$ gives
\[
    \Pr(e^*\notin\ALG)
    \le
    e^{-J/(ke)}+\bigl(1-e^{-J/(ke)}\bigr)e^{-J/e}
    \le
    e^{-J/(ke)}+e^{-J/e}
    \le
    2e^{-J/(ke)},
\]
where the last inequality uses $k\ge 1$. Hence
$\Pr(e^*\in\ALG)\ge 1-O(e^{-J/(ke)})$. Since this holds for every
$e^*\in\OPT$, the algorithm has the claimed probability-competitive ratio.

\end{proof}

\begin{corollary}
\label{cor:graphic}
For every graphic matroid and every $J \ge 1$, Algorithm~\ref{alg:ksparse} with
threshold $p = e^{-J/(2e)}$ achieves a probability-competitive guarantee of at
least $1 - O(e^{-J/(2e)})$.
\end{corollary}

\begin{proof}
A graphic matroid admits a $2$-column-sparse representation via its incidence
matrix, since each edge is incident to exactly two vertices. Moreover, under the
graphic-matroid input model from Section~\ref{sec:prelim}, each arriving edge is
revealed together with its endpoints, so its incidence column (and hence its
support) is known online.
Applying Theorem~\ref{thm:ksparse} with $k=2$ gives the result.
\end{proof}

\subsubsection{Capacitated probability-competitive ratio}
We now add the global capacity constraint $|\AUX|\le\lfloor\kappa r\rfloor$.
We couple the capacitated and uncapacitated executions and show that, when
$\kappa>J/(ke)$, the additional loss vanishes exponentially in the rank $r$.

\begin{corollary}
\label{cor:ksparse-kappa}
Fix $J\ge 1$ and set $\lambda:=J/(ke)$. Let $\kappa>\lambda$.
For every $k$-column-sparse matroid of rank $r$, Algorithm~\ref{alg:ksparse}
with capacity parameter $\kappa$ and threshold $p=e^{-\lambda}$ satisfies, for
every $v\in\OPT$,
\[
    \Pr(v\in\ALG)
    \ge
    1-O\bigl(e^{-\lambda}\bigr)
    -\exp\bigl(-r\,D(\kappa\|\lambda)\bigr),
\]
where $D(\kappa\|\lambda):=\kappa\ln(\kappa/\lambda)-\kappa+\lambda>0$.

In particular, for fixed $J$, $k$, and $\kappa>\lambda$, this bound approaches
$1-O(e^{-\lambda})$ as $r\to\infty$.
\end{corollary}

\begin{proof}
Let $\AUX^{\infty}$ and $\AUX^{\kappa}$ be the pools produced by the
uncapacitated and capacitated algorithms under the same arrivals.

By the labeling scheme, every improving arrival in $[p,1)$ receives a label in
$[r]$. Hence the number of improving arrivals in $[p,1)$ is $N_{[r]}[p,1)$, and
Proposition~\ref{prop:poisson-label-general} gives
$N_{[r]}[p,1)\sim\Po(r\lambda)$. Since every accepted element is improving,
$|\AUX^{\infty}|\le N_{[r]}[p,1)$.

Let
\[
    \calE:=\{N_{[r]}[p,1)\le \lfloor\kappa r\rfloor\}.
\]
On $\calE$, the uncapacitated execution never exceeds the global capacity, so
$\AUX^{\kappa}=\AUX^{\infty}$. By Observation~\ref{obs:union-vs-output}, for a
fixed $v\in\OPT$,
\[
    \Pr(v\in\OPT(\AUX^\kappa))
    \ge
    \Pr(v\in\OPT(\AUX^\infty)) - \Pr(\calE^c).
\]
Theorem~\ref{thm:ksparse} yields
$\Pr(v\in\OPT(\AUX^\infty))\ge 1-O(e^{-\lambda})$.

Finally, since $\kappa>\lambda$, the Cram\'er--Chernoff bound
(Lemma~\ref{lem:poisson-chernoff}(a)) gives
\[
    \Pr(\calE^c)
    =
    \Pr\bigl(\Po(r\lambda)>\lfloor\kappa r\rfloor\bigr)
    \le
    \exp\bigl(-D(\kappa r\|\lambda r)\bigr)
    =
    \exp\bigl(-rD(\kappa\|\lambda)\bigr),
\]
where $D(\kappa r\|\lambda r)=rD(\kappa\|\lambda)$ by definition. This proves the
claim.
\end{proof}

\begin{remark}
\label{rem:kappa-ksparse}
Set $\lambda=J/(ke)$. The concentration term in
Corollary~\ref{cor:ksparse-kappa} vanishes as $r\to\infty$ only when
$\kappa>\lambda$. If $\kappa<\lambda$, then
$\Pr(\Po(r\lambda)>\lfloor\kappa r\rfloor)\to 1$ as $r\to\infty$, so this
coupling argument yields no asymptotic guarantee. At the boundary
$\kappa=\lambda$, the same term does not vanish exponentially.

If $J<ke$, then $\lambda=J/(ke)<1$, so one may choose $\kappa<1$ while still
having $\kappa>\lambda$. In that case the global capacity $\lfloor\kappa r\rfloor$
is smaller than the rank $r$, but Corollary~\ref{cor:ksparse-kappa} still
recovers the uncapacitated guarantee as $r\to\infty$.
\end{remark}

\subsection{Union-based algorithms for laminar matroids}
\label{sec:laminar}

We now turn to laminar matroids. In this class we use an algorithm that
maintains one candidate pool feasible in $\calM^{(J)}$. The algorithm does not
choose online a partition of the accepted elements into $J$ independent sets.

This distinction matters for laminar constraints. If an element belongs to a
chain
$C_1\subset C_2\subset\cdots\subset C_h$, then assigning it to a track consumes
capacity in every set of this chain inside that track. Later elements may fail
to find one track with remaining capacity in all relevant sets, even when the
total capacity across the $J$ tracks is sufficient. The union-based algorithm
checks feasibility directly in $\calM^{(J)}$.

We first prove an uncapacitated guarantee with $\kappa=\infty$. We then add the
global capacity constraint. For $\kappa\ge J$ the constraint is redundant, and
for $\kappa\in(J/e,J)$ the same asymptotic guarantee is recovered as $r\to\infty$.

\subsubsection{The union-based algorithm}
\label{sec:laminar-union}

Let $\calM = (N, \calI)$ be a laminar matroid of rank $r$, defined by a laminar
family $\calL$ and capacities $\mu(C)$ for $C\in\calL$. We assume without loss of
generality that $N\in\calL$ and $\mu(N)=r$; adding this root constraint does not
change the matroid.

The $J$-fold union $\calM^{(J)}$ is the laminar matroid over the same family
$\calL$ with capacities $J\mu(C)$. Indeed, for laminar constraints, a set can be
partitioned into $J$ independent sets if and only if it satisfies
$|X\cap C|\le J\mu(C)$ for every $C\in\calL$.

We analyze a union-based greedy algorithm, $\Greedy^{(J)}(p)$, which maintains a
single independent set in $\calM^{(J)}$.

\begin{algorithm}[H]
\caption{$\Greedy^{(J)}(p)$ for laminar matroids}
\label{alg:greedy-union}
\DontPrintSemicolon
\KwData{A laminar matroid $\calM=(N,\calI)$ of rank $r$, an integer $J\ge 1$,
a threshold $p\in(0,1)$, a global capacity parameter
$\kappa\in(0,\infty]$, and elements arriving online with arrival times $t_e$.}
\KwResult{A feasible output for the $(J,\kappa)$-MSP.}

$\AUX\leftarrow\emptyset$\;
Set $K\leftarrow\infty$ if $\kappa=\infty$, and
$K\leftarrow\lfloor\kappa r\rfloor$ otherwise\;
Reject every element arriving before time $p$\;

\For{each element $e$ arriving at time $t_e\in[p,1)$, in increasing order of $t_e$}{
    \If{$e\in\OPT_{\calM}(N_{t_e})$ \textbf{and}
        $\AUX\cup\{e\}\in\calI^{(J)}$ \textbf{and} $|\AUX|<K$}{
        $\AUX\leftarrow \AUX\cup\{e\}$\;
    }
}
\Return{$\OPT_{\calM}(\AUX)$}\;
\end{algorithm}

The condition $e \in \OPT_{\calM}(N_{t_e})$ ensures that elements are improving
with respect to the original matroid $\calM$.

\subsubsection{Labeling scheme and chain conflicts}

We use the chain-order labeling scheme of B\'erczi, Livanos, Soto, and Verdugo~\cite{BercziLSV24}. Fix
$e^*\in\OPT_{\calM}(N)$, and let $C_1\subset C_2\subset\cdots\subset C_h$ be the
maximal chain of laminar sets containing $e^*$. We choose a total order $\pi$ of
the ground set such that $e^*$ is first and, for every chain set $C$ containing
$e^*$, all elements contained in $C$ precede all elements outside $C$. Dummy
elements are ordered subject to the same nesting rule. The induced labeling
assigns to each element of $\OPT_t$ its relative rank in this order. Therefore,
for every such chain set $C$ and every improving time $t$, each element of
$\OPT_t\cap C$ receives a label at most $\mu(C)$.

This is a labeling scheme in the sense of Section~\ref{sec:prelim}: the order
$\pi$ is fixed, and at each improving time $t$ the labels are determined by the
relative order of the elements of $\OPT_t$ under $\pi$. Thus the assignment is
independent of the internal arrival order inside $N_t$.

\begin{definition}[$J$-well-indexed words]
Let $z$ be the sequence of labels of all improving arrivals in $[p,1)$, read in
reverse chronological order. A word $z\in\{1,\dots,r\}^*$ is
\emph{$J$-well-indexed} if it factors as $z=x1y$, where $x$ does not contain the
symbol $1$, and for every $c\in\{1,\dots,r\}$, the suffix $y$ satisfies
\[
    |\{i:y_i\le c\}|\le Jc-1.
\]
\end{definition}

\begin{lemma}
\label{lem:well-indexed-accepts}
Consider Algorithm~\ref{alg:greedy-union} with capacity parameter
$\kappa=\infty$. If the improving word $z$ of $e^*$ is $J$-well-indexed, then
Algorithm~\ref{alg:greedy-union} accepts $e^*$.
\end{lemma}

\begin{proof}

Since $z$ factors as $x1y$ with $x$ containing no symbol $1$, the distinguished
symbol $1$ is the arrival of $e^*$. Indeed, by Proposition~\ref{prop:folklore},
$e^*\in\OPT_t$ for every $t\ge t_{e^*}$. Since $e^*$ is first in the order $\pi$,
it has label $1$ at all such times. Hence no later improving arrival can carry
label $1$. Thus $e^*$ arrives in $[p,1)$, and the
suffix $y$ corresponds to improving elements arriving in $[p,t_{e^*})$. Run
Algorithm~\ref{alg:greedy-union} with capacity parameter $\kappa=\infty$, and
let $\AUX^-$ be its candidate pool immediately before $t_{e^*}$.

If $e^*$ is rejected, then the condition
$\AUX^-\cup\{e^*\}\in\calI^{(J)}$ is false. Hence there exists a set $C\in\calL$
containing $e^*$ such that $|\AUX^-\cap C|=J\mu(C)$.
Let $c = \mu(C)$. By the
labeling scheme, every improving element in $C$ receives a label in $\{1, \dots,
c\}$. Thus, every element in $\AUX^- \cap C$ contributes a symbol $\le c$ to $y$.
This implies $|\{i : y_i \le c\}| \ge Jc$, which contradicts the assumption that
$z$ is $J$-well-indexed.
\end{proof}

\subsubsection{Uncapacitated probability-competitive ratio}

To lower bound the success probability, we analyze the discrete event that $z$ is
$J$-well-indexed. Conditional on the arrival time $t\in[p,1)$ of $e^*$, let
$N_c$ be the number of improving arrivals in $[p,t)$ with label at most $c$.
Then $N_c\sim\Po(c\ln(t/p))$. For the labeling fixed above, the condition that
$x$ contains no symbol $1$ is automatic once $e^*$ arrives. Hence, conditional on
$t\in[p,1)$, the word $z$ is $J$-well-indexed whenever $N_c\le Jc-1$ for all
$c\in\{1,\dots,r\}$.

\begin{theorem}
\label{thm:laminar-main}
For every laminar matroid $\calM$ and every $J\ge 1$,
Algorithm~\ref{alg:greedy-union} with capacity parameter $\kappa=\infty$ and
threshold $p=e^{-J/e}$ selects every $e^*\in\OPT_{\calM}(N)$ with probability
at least
\[
1-e^{-J/e}-\frac{e^{-J/e}}{1-e^{-J/e}}.
\]
In particular, its probability-competitive ratio is $1-O(e^{-J/e})$.
Moreover, the lower bound obtained from the $J$-well-indexed event is nondecreasing
in $J$ after optimizing over $p$.
\end{theorem}

\begin{proof}
Fix an element $e^*\in\OPT_{\calM}(N)$ and its arrival time $t\in[p,1)$. Let
$N_c$ be the number of improving arrivals in $[p,t)$ with label at most $c$.
Conditional on $t$, $N_c$ has a Poisson distribution with parameter
$c\lambda$, where $\lambda=\ln(t/p)$. Since $t\in[p,1)$ and $p=e^{-J/e}$, we
have $\lambda\in[0,J/e]$. By Lemma~\ref{lem:well-indexed-accepts}, the algorithm
accepts $e^*$ if the corresponding word is $J$-well-indexed. Therefore, if the
algorithm rejects $e^*$ then $N_c \ge Jc$ for some $c \in \{1,
\dots, r\}$.

To show that the guarantee is monotonically non-decreasing in $J$, we use a
coupling argument. Since $Jc - 1 \le (J+1)c - 1$ for every $c \ge 1$, any
sequence of arrivals that is $J$-well-indexed is deterministically
$(J+1)$-well-indexed. Thus, for any fixed $p$, the lower bound on the success
probability derived from the well-indexed event for $\Greedy^{(J+1)}(p)$ is
pointwise bounded below by that for $\Greedy^{(J)}(p)$. Taking the supremum over
$p$ preserves this monotonicity, ensuring that the analytical guarantee improves as
$J$ increases.

For the explicit bound, fix $p = e^{-J/e}$. The algorithm can reject $e^*$ only
if there exists an index $c \in \{1, \dots, r\}$ such that $N_c \ge Jc$.
Fix the arrival time $t\in[p,1)$ of $e^*$, and recall that
$N_c\sim\Po(c\lambda)$ with $\lambda=\ln(t/p)\in[0,J/e]$. By the union bound,
\[
\Pr(e^*\notin\ALG\mid t)\le \sum_{c=1}^r \Pr(N_c \ge Jc\mid t).
\]
The Chernoff bound for Poisson variables (Lemma~\ref{lem:poisson-chernoff}(a)) gives
\[
\Pr(N_c \ge Jc\mid t) \le \exp\big(-c (J \ln(J/\lambda) + \lambda - J)\big).
\]
The exponent $J\ln(J/\lambda)+\lambda-J$ is strictly decreasing in $\lambda$ on
$[0,J/e]$. Therefore, for every $t\in[p,1)$,
\[
\Pr(N_c \ge Jc\mid t) \le \exp\big(-c(J \ln(e) + J/e - J)\big) = \exp(-c J/e) = p^c.
\]
Summing over $c\ge 1$ yields
\[
\Pr(e^*\notin\ALG\mid t)\le \sum_{c=1}^r p^c \le \frac{p}{1-p}.
\]
Finally, $\Pr(t_{e^*}<p)=p$, and therefore
\[
\Pr(e^*\notin\ALG) \le \Pr(t_{e^*}<p) + \Pr(e^*\notin\ALG\mid t_{e^*}\in[p,1))
\le p+\frac{p}{1-p}.
\]
With $p=e^{-J/e}$, this gives the stated bound.
\end{proof}

The explicit lower bound is mainly useful in the large-$J$ regime. For small
values of $J$, and in particular for $J=1$, sharper threshold choices give
better bounds.

\begin{remark}
For $J=1$, Algorithm~\ref{alg:greedy-union} is the Greedy-Improving algorithm
of B\'erczi, Livanos, Soto, and Verdugo~\cite{BercziLSV24} for laminar
matroids. They proved that its optimal probability guarantee is exactly
$1-\ln 2$, attained with threshold $p=1/2$. The union-bound estimate in
Theorem~\ref{thm:laminar-main} does not recover this value, since it uses
$p=e^{-1/e}$. We use that estimate only to obtain a simple guarantee for
general $J$.
\end{remark}

\subsubsection{Capacitated probability-competitive ratio}

We now add the global capacity constraint $|\AUX|\le\lfloor\kappa r\rfloor$.
If $\kappa\ge J$, this constraint is redundant: since $N\in\calL$ and $\mu(N)=r$,
membership in $\calI^{(J)}$ already implies $|\AUX|\le Jr$. We therefore assume
$\kappa<J$.

The proof below compares the execution with capacity parameter $\kappa$ to the
execution with $\kappa=\infty$. With the threshold $p=e^{-J/e}$, the total
number of improving arrivals in $[p,1)$ has distribution $\Po(rJ/e)$. Therefore,
whenever $\kappa>J/e$ the additional loss due to the global capacity constraint
is exponentially small in the rank $r$.

\begin{corollary}
\label{cor:laminar-kappa}
Fix $J\ge 1$ and $\kappa\in(J/e,J)$. For every laminar matroid $\calM$ of
rank $r$, Algorithm~\ref{alg:greedy-union} with threshold $p=e^{-J/e}$ and
capacity parameter $\kappa$ satisfies, for every $e^*\in\OPT$,
\[
    \Pr(e^*\in\ALG)
    \ge
    1-O(e^{-J/e})
    -
    \exp\left(
        -r\,D\left(\kappa\,\middle\|\,\frac{J}{e}\right)
    \right),
\]
where $D(\kappa\|\lambda):=\kappa\ln(\kappa/\lambda)-\kappa+\lambda$.
Consequently, for fixed $J$ and $\kappa\in(J/e,J)$, the
probability-competitive ratio converges to $1-O(e^{-J/e})$ as $r\to\infty$.
\end{corollary}

\begin{proof}

Set $\lambda=J/e$, so $p=e^{-\lambda}$. Run the executions with parameters
$\kappa$ and $\infty$ on the same arrivals, weights, and tie-breaking rules.
Let $\AUX^\kappa$ and $\AUX^\infty$ be their candidate pools, and let
$\ALG^\kappa=\OPT_{\calM}(\AUX^\kappa)$ and
$\ALG^\infty=\OPT_{\calM}(\AUX^\infty)$.

Under the labeling scheme fixed for $e^*$, every improving arrival in $[p,1)$
receives a label in $[r]$. Hence the number of improving arrivals in $[p,1)$ is
$N_{[r]}[p,1)$. By Proposition~\ref{prop:poisson-label-general},
\[
    N_{[r]}[p,1)\sim\Po(r\lambda).
\]
Define
\[
    \calE
    :=
    \{N_{[r]}[p,1)\le \lfloor\kappa r\rfloor\}.
\]
On $\calE$, the final size of the uncapacitated candidate pool is at most
$\lfloor\kappa r\rfloor$. Since this size is nondecreasing during the
execution, the uncapacitated execution never violates the global capacity
bound. Thus the capacitated and uncapacitated executions add the same elements
to their candidate pools on $\calE$.

By Observation~\ref{obs:union-vs-output},
\[
    \Pr(e^*\in\ALG^\kappa)
    \ge
    \Pr(e^*\in\ALG^\infty)-\Pr(\calE^c).
\]
Theorem~\ref{thm:laminar-main} gives
$\Pr(e^*\in\ALG^\infty)\ge 1-O(e^{-J/e})$.
It remains to bound $\Pr(\calE^c)$. Since $\kappa>\lambda$, the
Cram\'er--Chernoff bound for Poisson variables (Lemma~\ref{lem:poisson-chernoff}(a)) gives
\[
    \Pr(\calE^c)
    =
    \Pr(\Po(r\lambda)>\lfloor\kappa r\rfloor)
    \le
    \exp(-rD(\kappa\|\lambda)).
\]
Substituting $\lambda=J/e$ proves the claim.\qedhere
\end{proof}

\section{Conclusion}

The main result of this paper is exact optimality for the $J$-MSP on
transversal matroids. For every fixed $J$, the optimal probability guarantee is
the optimal success probability of the classical $J$-choice secretary problem of
Gilbert and Mosteller~\cite{GilbertMosteller}. Thus the rank-one transversal
matroid is already a tight worst case for the whole transversal class.

The proof is local around each fixed element $v\in\OPT$. At the arrival time of
$v$, let $f_v$ be the right vertex matched to $v$ in the canonical matching.
Before $v$ arrives, the only relevant conflicts are improving arrivals whose
canonical partner is $f_v$. These arrivals form one label process. This process
has the same law as the record process in the $J$-choice secretary problem. The
Gilbert--Mosteller thresholds therefore transfer to Algorithm~\ref{alg:optJ}
without loss.

The capacitated transversal setting shows what changes when each right vertex
has capacity $b$ and the pool has global capacity $\kappa r$. For a fixed
optimal element, the obstruction is now a block of $b$ labels. The
single-threshold routing algorithm of Section~\ref{sec:routing} gives explicit
bounds in terms of $J$, $b$, $q$, and $\kappa$. When $q\to\infty$, the loss from
the global capacity constraint disappears if the expected uncapacitated load is
below the capacity threshold. For $b=1$ and $q\to\infty$, the optimized limiting
guarantee converges to $1-e^{-\kappa}$ as $J\to\infty$ for fixed $\kappa$.

The $k$-column-sparse and laminar examples show two different certificates for
the same model. The $k$-column-sparse algorithm is multi-track. It maintains
explicit sets $F_1,\ldots,F_J$ and proves feasibility track by track. Its
analysis must show that each accepted element can be routed online to a track
that is not blocked. The laminar algorithm is union-based. It maintains one pool
and enforces feasibility in $\calM^{(J)}$ directly. This avoids online routing
and avoids splitting laminar capacities across tracks. It shifts the burden to
enforcing the union constraint during the online execution.

For multi-track algorithms, a natural broader framework is given by the
$k$-directed certifiers of Marinkovic, Soto, and
Verdugo~\cite{MarinkovicSV26}. In that framework, each accepted element has a
certificate. A directed blocking relation between certificates is used to prove
feasibility. The parameter $k$ bounds how many elements of one independent set
can block a fixed certificate. The multi-track algorithm of
Section~\ref{sec:ksparse} follows this logic. An improving element is accepted
when some track contains no previously accepted element that blocks it. The
analysis works because the blockers of a fixed element of $\OPT$ are controlled
by few labels. Extending the labeling-scheme analysis to general $k$-directed
certifiers is a natural direction that we do not pursue here.

\paragraph*{Acknowledgements} This work was partially supported by ANID (Agencia Nacional de Investigaci\'on y Desarrollo, Chile) through FONDECYT Grant No.~1231669, ANID BECAS/DOCTORADO NACIONAL 21252789 and the BASAL Center for Mathematical Modeling (FB210005).
\bibliographystyle{plain}
\bibliography{transversal_mmsp_refs}

\end{document}